\def\ARXIVVERSION{1}
\documentclass[onecolumn,draftcls]{IEEEtran}

\usepackage{cite}
\usepackage{pbox}
\usepackage{color}
\usepackage{xcolor}
\usepackage{graphicx}
\usepackage{amsmath}
\usepackage{amssymb}
\usepackage{amsthm}
\usepackage{mathrsfs}
\usepackage{mathtools}
\usepackage{psfrag}
\usepackage[caption=false,font=footnotesize]{subfig}
\usepackage{tikz,pgfplots}
\usepackage{url}
\usepackage{booktabs}
\usepackage{enumitem}
\usepackage[hidelinks]{hyperref}
\pgfplotsset{compat=newest}
\usetikzlibrary{matrix,positioning,calc,fit,backgrounds,arrows.meta}

\let\oldbrace\{
\def\{{\oldbrace\kern0.5pt}

\def\Var{\mathop{\rm Var}\nolimits}%
\newcommand{\Ac}{\mathcal{A}}

\newcommand{\Pc}{\mathcal{P}}

\newcommand{\Sc}{\mathcal{S}}

\newcommand{\Uc}{\mathcal{U}}
\newcommand{\Vc}{\mathcal{V}}

\newcommand{\Xc}{\mathcal{X}}

\newcommand{\Ut}{{\tilde{U}}}

\newcommand{\xt}{{\tilde{x}}}

\DeclareMathOperator\E{\sf E}

\newcommand{\N}{\mathrm{N}}
\newcommand{\U}{\mathrm{Unif}}

\newcommand{\ind}{\boldsymbol{1}}

\definecolor{lineblue}{rgb}{.2 .6 .9}
\definecolor{linered}{rgb}{.9 .3 .3}
\definecolor{linegreen}{rgb}{.3 .8 .3}
\definecolor{linepurple}{rgb}{.75 .5 .75}
\definecolor{LineGray}{rgb}{.6 .6 .6}
\definecolor{AxisGray}{rgb}{.3 .3 .3}
\definecolor{LineBlue}{rgb}{.2 .6 .9}
\definecolor{LineRed}{rgb}{.9 .3 .3}
\definecolor{LineGreen}{rgb}{.3 .8 .3}
\definecolor{LinePurple}{rgb}{.8 .3 .8}

\newcommand{\TV}{\mathsf{TV}}
\newcommand{\full}{\mathrm{full}}
\newcommand{\ro}{\mathrm{ro}}
\newcommand{\ri}{\mathrm{ri}}

\newcommand{\conv}{\operatorname{conv}}
\newcommand{\cl}{\operatorname{cl}}
\newcommand{\Fix}{\operatorname{Fix}}
\newcommand{\suppo}{\operatorname{supp}}

\newcommand{\Rfull}{R_{\full}}
\newcommand{\Rro}{R_{\ro}}
\newcommand{\Rri}{R_{\ri}}
\newcommand{\Gammaro}{\Gamma_{W_0}}
\newcommand{\Deltaro}{\Delta_{W_0}}

\newcommand{\Cren}{\mathcal C(W_0)}
\newcommand{\Xsupp}{\Xc_0}
\newcommand{\baseG}{\mathsf G_0}
\newcommand{\baseC}{\mathsf C_0}

\newcommand{\preceqin}{\preceq_{\rm in}}

\newtheorem{theorem}{Theorem}
\newtheorem{lemma}{Lemma}
\newtheorem{corollary}{Corollary}
\newtheorem{proposition}{Proposition}

\theoremstyle{definition}
\newtheorem{example}{Example}
\newtheorem{remark}{Remark}
\newtheorem{definition}{Definition}

\IEEEoverridecommandlockouts
\allowdisplaybreaks

\begin{document}

\title{Semantic Transferability of Generative Rate–Distortion–Perception Systems}

\newcommand{\authorIEEE}[1]{\ifdefined\ARXIVVERSION\else,~\IEEEmembership{#1}\fi}
\author{{Jinwan~Jeon, Seoyeon~Hwang, and~Sung~Hoon~Lim}%
\thanks{{J. Jeon, S. Hwang, and S. H. Lim are with Hallym University, Chuncheon-si, Gangwon-do, South Korea (e-mail: jeonjw@hallym.ac.kr; hsy@hallym.ac.kr; shlim@hallym.ac.kr).}}}

\maketitle

\begin{abstract}
Generative lossy compression combines a rate-limited representation with a stochastic decoder that preserves source-dependent information while generating unresolved detail. We ask what remains reusable when a rate--distortion--perception (RDP) system selected for one operating point is transferred to a prescribed conditional-generation task. We distinguish the single-letter RDP factorization from its block-code implementation and compare Full Redesign with three increasingly restrictive transfer strategies. Renderer-Only Transfer retains only the renderer and redesigns semantic steering. Representation-Induced Transfer additionally retains the semantic latent law induced by the original representation. Full-Model Transfer retains a compatible RDP code, marginalizing over the task conditional law to form its semantic encoder. For finite alphabets, we characterize Renderer-Only Transfer's exact compression--common-randomness rate region, renderer-hull controllability condition, and compression-rate penalty relative to Full Redesign. Under unlimited common randomness, this penalty is the minimum residual semantic information $I(S;U|X)$ forced by the inherited renderer. Representation-Induced Transfer is exact if and only if the semantic condition depends on the source only through the recurrent-class label $C$ of the end-to-end RDP test channel. Every compatible RDP code reused under Full-Model Transfer preserves $C$ and supports every task satisfying $X-C-S$. Any additionally transferable within-class semantic dependence must satisfy $I(S;X|C)\le R_0-H(C)$, where $R_0$ is the inherited compression rate. The bound is exact when $R_0=H(C)$. A scalar Gaussian specialization yields a renderer-noise controllability threshold, a closed-form transfer penalty, and examples of identical Gaussian RDP test channels with different Renderer-Only Transfer capabilities.
\end{abstract}


\section{Introduction}
\label{sec:introduction}

Learned compression systems increasingly use stochastic or generative decoders to reconstruct signals at rates for which many source details cannot be transmitted explicitly.  In this regime, the decoder has two jobs.  It must preserve the information that is tied to the particular source realization, and it must generate the unresolved detail in a statistically plausible manner.  This viewpoint appears throughout modern learned image compression~\cite{BalleEtAl2017,TheisEtAl2017,BalleEtAl2018,MinnenEtAl2018,AgustssonTheis2020} and is especially explicit in generative compression systems based on adversarial, likelihood-based, and diffusion-inspired decoders~\cite{Tschannen2018,Agustsson2019,Mentzer2020,YanEtAl2021,YanEtAl2022,MuckleyEtAl2023,CareilEtAl2024}.  Similar encoder--latent--generator factorizations are also standard in latent-variable generative modeling~\cite{KingmaWelling2014,HoEtAl2020,RombachEtAl2022}.

The rate--distortion--perception (RDP) framework supplies an information-theoretic abstraction of this behavior.  Distortion measures how well a reconstruction remains tied to the particular source sample, whereas perception constrains the distribution of reconstructed samples to agree with the source distribution~\cite{BlauMichaeli2018,BlauMichaeli2019}.  Coding theorems and refinements of the framework have clarified the roles of deterministic, private-randomness, and common-randomness codes, as well as the distinction between weak and strong distributional criteria~\cite{TheisWagner2021,ChenEtAl2022,Wagner2022,HamdiWagnerGunduz2026,SalehkalaibarEtAl2024}.  
Throughout this paper, we describe this statistical structure through the single-letter RDP factorization $q_0(x, u, \xt)$ such that,
\begin{align}
    X \longrightarrow U \xrightarrow{\;W_0\;} \tilde X,
    \label{eq:intro-base-chain}
\end{align}
where $q_0(x)$ is the source distribution, $q_0(u|x)$ is a single-letter representation law, and $W_0(\tilde x|u):=q_0(\xt|u)$ is a stochastic decoder, which we call the \emph{renderer}. 
{This factorization provides an information-theoretic abstraction of the encoder--latent--generator structure used in generative compression. It captures how the source information retained in $U$ guides the generation of $\tilde X$ through $W_0$, with distortion and perception constraining the fidelity and distribution of the reconstruction, respectively.}

{
The practical motivation of this work is model transfer. For example, a future transfer task may require samples conditioned on a semantic variable $S$, such as a class label or an attribute. We describe this task by a prescribed transfer target distribution $q(s,x)=q(s)q(x|s)$ and focus on the matched-marginal setting
$q(x)=q_0(x)$.
Thus, the desired output marginal remains unchanged, while the generated output $\hat X$ must have the prescribed dependence on $S$. Since the base RDP system already provides a latent-to-output renderer $W_0$, a natural question is whether this renderer can be reused through a suitable mapping from $S$ to its latent input:
\begin{align}
    S \longrightarrow \Ut \xrightarrow{\;W_0\;} \hat X.
    \label{eq:intro-transfer-chain}
\end{align}
The key issue is that correct unconditional generation, i.e., generation that matches only the marginal $q(x)$, does not by itself guarantee that the transferred renderer $W_0$ can realize each specific transfer task $q(x|s)$. Thus, our primary question is: Can that renderer be transferred to a later conditional-generation task by changing how the condition accesses its latent input?
}

Indeed, training a new conditional generator for every later attribute, task, or semantic description can be expensive, and several generative-modeling methods instead retain an already trained generator and learn a new latent constraint or steering mechanism~\cite{EngelEtAl2018}.  Modern diffusion systems provide a particularly visible example: ControlNet and T2I-Adapter keep a large pretrained text-to-image backbone fixed while learning condition-specific modules that inject structural information such as edges, depth, segmentation, or pose~\cite{ZhangRaoAgrawala2023,MouEtAl2024}.  These architectures are practical motivations rather than operational instances of the coding model developed here; the shared structural question is what a previously trained generative component can support when new conditions arrive later. {A related fixed-decoder constraint arises in prompt compression for black-box language models, where the compressor is optimized while the language model remains unchanged~\cite{NagleEtAl2024}.} {An information-theoretic formulation of fixed-generator reuse also appears in latent-space watermarking~\cite{JeonLeeLim2026}, where a recoverable message is embedded while preserving a prescribed conditional output distribution.}

The distributional criteria on the base and semantic sides are chosen to be directly comparable.  Strong near-perfect perception requires
\begin{align}
    \left\|p(\tilde x^n)-q_0(x^n)\right\|_{\TV}\to0,
    \label{eq:intro-strong-perception}
\end{align}
whereas the later conditional generator is required to satisfy
\begin{align}
    \left\|p(s^n,\hat x^n)-\prod_{i=1}^n q(s_i,\hat x_i)\right\|_{\TV}\to0.
    \label{eq:intro-strong-synthesis}
\end{align}
Thus, channel synthesis is not introduced as an unrelated downstream metric.  It is the conditional version of the same block-distribution matching principle{~\cite{HamdiGunduz2026Survey}}: the first condition asks for realistic unconditional samples, while the second asks for realistic samples with the correct relationship to $S$.

We use Full Redesign, which allows both the encoder and decoder to be designed freely for the conditional-generation task, as a performance benchmark and compare it with three \emph{transfer strategies} that retain different components of the pretrained system. Their statistical and operational structures are summarized in Fig.~\ref{fig:three-levels}, with formal definitions postponed to Section~\ref{sec:problem}. Renderer-Only Transfer retains only the generative renderer $W_0$ and is free to learn a new condition-to-latent steering law. Representation-Induced Transfer additionally uses the semantic latent law induced by the pretrained representation $q_0(u|x)$, which models using the class- or condition-dependent latent statistics already present in the base representation rather than learning a new steering law. Full-Model Transfer is stronger still at the operational level: it reuses an already constructed RDP block code, including its encoder, latent codebook, common-randomness mechanism, and renderer.

\begin{figure}[t]
\centering
\resizebox{\textwidth}{!}{%
\begin{tikzpicture}[
    font=\small,
    block/.style={
        draw,
        rectangle,
        minimum width=27mm,
        minimum height=11mm,
        text width=23mm,
        align=center,
        inner sep=2pt
    },
    arrow/.style={-{Latex[length=1.45mm,width=1.0mm]}, line width=0.32pt},
    alab/.style={midway, above=0.8mm, inner sep=0pt},
    rowlab/.style={anchor=east, align=right, font=\small}
]

\def\dx{4.05}
\def\xleft{-2.70}
\def\xrighttwo{6.75}
\def\xrightthree{10.80}
\def\xlabel{-3.05}

\node[rowlab] at (\xlabel,0) {Base RDP Model};
\node[block] (benc) at (0,0) {Representation\\$q_0(u|x)$};
\node[block] (brend) at (\dx,0) {Fixed Renderer\\$W_0$};
\coordinate (bin) at (\xleft,0);
\coordinate (bout) at (\xrighttwo,0);
\draw[arrow] (bin) -- node[alab] {$X$} (benc.west);
\draw[arrow] (benc.east) -- node[alab] {$U$} (brend.west);
\draw[arrow] (brend.east) -- node[alab] {$\tilde X$} (bout);

\node[rowlab] at (\xlabel,-1.80) {Renderer-Only Transfer};
\node[block] (renc) at (0,-1.80) {Semantic Map\\$p(v|s)$};
\node[block] (rmap) at (\dx,-1.80) {Latent Map\\$p(u|v)$};
\node[block] (rrend) at ({2*\dx},-1.80) {Fixed Renderer\\$W_0$};
\coordinate (rin) at (\xleft,-1.80);
\coordinate (rout) at (\xrightthree,-1.80);
\draw[arrow] (rin) -- node[alab] {$S$} (renc.west);
\draw[arrow] (renc.east) -- node[alab] {$V$} (rmap.west);
\draw[arrow] (rmap.east) -- node[alab] {$U$} (rrend.west);
\draw[arrow] (rrend.east) -- node[alab] {$\hat X$} (rout);

\node[rowlab] at (\xlabel,-3.60) {Representation-Induced Transfer};
\node[block] (ienc) at (0,-3.60) {Induced Map\\$B_0(u|s)$};
\node[block] (irend) at (\dx,-3.60) {Fixed Renderer\\$W_0$};
\coordinate (iin) at (\xleft,-3.60);
\coordinate (iout) at (\xrighttwo,-3.60);
\draw[arrow] (iin) -- node[alab] {$S$} (ienc.west);
\draw[arrow] (ienc.east) -- node[alab] {$U$} (irend.west);
\draw[arrow] (irend.east) -- node[alab] {$\hat X$} (iout);

\node[rowlab] at (\xlabel,-5.40) {Full-Model Transfer};
\node[block] (fsamp) at (0,-5.40) {Task Sampler\\$q(x^n|s^n)$};
\node[block] (fcode) at (\dx,-5.40) {Fixed RDP Code\\$\mathsf C_{0,n}$};
\coordinate (fmin) at (\xleft,-5.40);
\coordinate (fmout) at (\xrighttwo,-5.40);
\draw[arrow] (fmin) -- node[alab] {$S^n$} (fsamp.west);
\draw[arrow] (fsamp.east) -- node[alab] {$X^n$} (fcode.west);
\draw[arrow] (fcode.east) -- node[alab] {$\tilde X^n$} (fmout);

\node[rowlab] at (\xlabel,-7.20) {Full Redesign};
\node[block] (fenc) at (0,-7.20) {Semantic Map\\$p(v|s)$};
\node[block] (fgen) at (\dx,-7.20) {Generator\\$p(x|v)$};
\coordinate (fin) at (\xleft,-7.20);
\coordinate (fout) at (\xrighttwo,-7.20);
\draw[arrow] (fin) -- node[alab] {$S$} (fenc.west);
\draw[arrow] (fenc.east) -- node[alab] {$V$} (fgen.west);
\draw[arrow] (fgen.east) -- node[alab] {$\hat X$} (fout);

\coordinate (sepwest) at (current bounding box.west);
\coordinate (sepeast) at (current bounding box.east);
\draw[dashed, line width=0.32pt] (sepwest |- 0,-0.90) -- (sepeast |- 0,-0.90);
\draw[dashed, line width=0.32pt] (sepwest |- 0,-6.30) -- (sepeast |- 0,-6.30);

\draw[arrow, line width=0.6pt, -{Latex[length=1.8mm,width=1.25mm]}]
    ([xshift=3mm]sepwest |- 0,-0.25)
    -- node[midway, right=1.5mm, yshift=3mm, inner sep=0pt] {Transfer}
    ([xshift=3mm]sepwest |- 0,-1.55);

\draw[arrow, line width=0.6pt, -{Latex[length=1.8mm,width=1.25mm]}]
    ([xshift=3mm]sepwest |- 0,-6.95)
    -- node[midway, right=1.5mm, yshift=-3mm, inner sep=0pt] {Benchmark}
    ([xshift=3mm]sepwest |- 0,-5.65);

\end{tikzpicture}%
}
\caption{Structural Comparison of the Base RDP Model, Full Redesign, and the Three Transfer Strategies. {All rows except Full-Model Transfer show} the corresponding single-letter statistical structures. {The Full-Model Transfer row shows the equivalent representation~\eqref{eq:fm-induced-law}. Marginalizing $X^n$ gives the encoder in~\eqref{eq:fm-encoder}.}}
\label{fig:three-levels}
\end{figure}
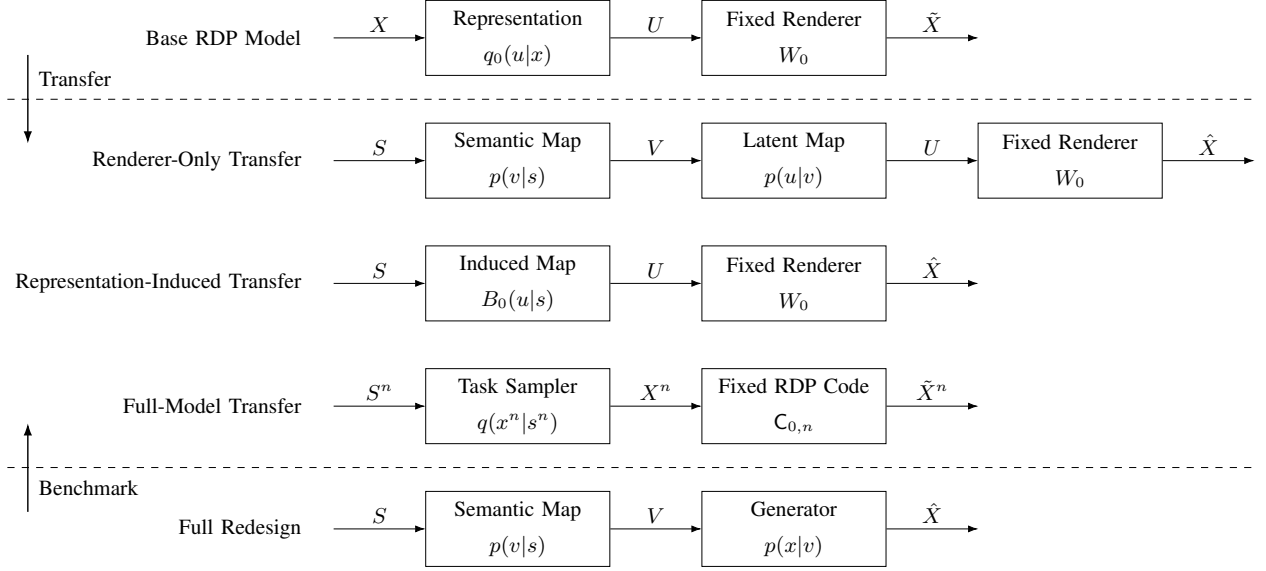

The first three semantic systems admit a direct rate-region comparison. Whenever Representation-Induced Transfer exactly realizes the desired task,
\begin{align}
    \mathscr R_{\ri}(q|\baseG)
    \subseteq
    \mathscr R_{\ro}(q|W_0)
    \subseteq
    \mathscr R_{\full}(q).
    \label{eq:intro-hierarchy}
\end{align}
The inclusions reflect what is allowed to change: Full Redesign chooses the entire semantic generator, Renderer-Only Transfer fixes $W_0$, and Representation-Induced Transfer also fixes the semantic latent statistics inherited from $q_0(u|x)$. Full-Model Transfer is a different, operationally stronger question because it fixes an actual block-code sequence rather than merely a single-letter target law.

This separation exposes three questions that are easy to conflate in practice. First, is the old \emph{generator} $W_0$ capable of the new task after latent steering is adapted? Second, if it is, does the old \emph{representation} already induce the correct semantic latent statistics? Third, does an actual deployed RDP code inherit this single-letter behavior when used without redesign? The first two questions are controlled by different objects. The convex renderer hull
\begin{align}
    \mathcal C(W_0)
    :=\conv\{W_0(x|u):u\in\Uc\}
    \label{eq:intro-hull}
\end{align}
characterizes what becomes available after steering adaptation. In contrast, the stationary distributions of the end-to-end RDP test channel
\begin{align}
    T_0(\tilde x|x)
    :=q_0(\tilde x|x)=\sum_uq_0(u|x)W_0(\tilde x|u)
    \label{eq:intro-t0}
\end{align}
characterize what is preserved by Representation-Induced Transfer. 
Full-Model Transfer then asks what an actual compatible RDP code preserves when reused without {\em any} redesign. 
{We identify a family of semantic tasks that every compatible RDP code is guaranteed to support. A particular implementation may support additional tasks, but its compression rate limits how much additional semantic dependence it can preserve. The choice of which component to retain also matters: two RDP systems with the same source-to-reconstruction channel can support different sets of transfer tasks when only their renderers are reused.}

\paragraph*{Contributions.}
In the present work, we make the following main contributions.
\begin{enumerate}[leftmargin=*,label=\arabic*)]

    \item \emph{Renderer-Only Transfer.} {We characterize which conditional-generation tasks a retained renderer can realize and the exact tradeoff between the compression rate and the common-randomness rate required for their synthesis (Theorem~\ref{thm:reuse-region} and Corollary~\ref{cor:controllability}). This separates tasks that are impossible through the renderer from those that are feasible but require a higher compression rate than Full Redesign. Even with unlimited common randomness, this rate penalty can remain strictly positive. We show that it equals the minimum information about the semantic condition that the renderer's latent input must carry beyond what is present in the generated output (Theorem~\ref{thm:unlimited-reuse}). We also show that a renderer supports every transfer task with the same output marginal as the base system if and only if it attains the Full Redesign rate region for every such task (Theorem~\ref{thm:universal-no-loss}).}

    \item \emph{Representation-Induced Transfer.} {We characterize the semantic tasks that can be transferred when the pretrained representation law $q_0(u|x)$ is retained in addition to the renderer $W_0$. The key object is the recurrent-class label $C$, which records the source distinctions preserved without error by the base RDP system. A task is exactly transferable if and only if its semantic condition depends on the source only through $C$ (Theorem~\ref{thm:maximal-invariant}). Thus, $C$ determines the entire family of exactly transferable tasks and provides a minimal interface supporting all of them (Corollary~\ref{cor:minimal-universal-interface}). Moreover, whenever transfer is exact, the inherited latent statistics attain the minimum compression rate of Full Redesign under unlimited common randomness. Hence, redesigning the semantic steering cannot further reduce that rate (Theorem~\ref{thm:ri-optimal}).}

    \item \emph{Full-Model Transfer.} {While Representation-Induced Transfer designs a new semantic synthesis code for the inherited latent statistics, Full-Model Transfer reuses an existing RDP code at its original rates. Thus, the compression rate is fixed by the base system rather than optimized for the specific transfer task. We show that every compatible implementation preserves the recurrent-class label $C$, which guarantees successful transfer for matched-marginal transfer targets $(S,X)$ satisfying $X-C-S$. We also derive a converse that limits the semantic dependence beyond $C$ that a particular implementation can preserve, in terms of its inherited compression rate (Theorem~\ref{thm:operational-transfer}).}

    \item \emph{Separations and Gaussian Insight.} {We illustrate the implications of our finite-alphabet characterizations through discrete examples and extend the analysis to a scalar Gaussian model. In both settings, we construct two RDP factorizations with the same base source--reconstruction joint law for which a transfer task is impossible through one renderer but achievable through the other (Propositions~\ref{prop:same-T0-separation} and~\ref{prop:gaussian-same-test}). Thus, the retained renderer determines which transfer targets can be realized and at what compression rates, but this information is not determined by the end-to-end RDP test channel alone. The scalar Gaussian model makes the renderer's role explicit through its noise variance. Under unlimited common randomness, transfer at a finite compression rate is possible exactly when the target conditional variance strictly exceeds the variance introduced by the renderer. We also derive the corresponding minimum compression rate and its penalty relative to Full Redesign (Proposition~\ref{prop:gaussian-reuse}).}
\end{enumerate}

{In summary, the rate-region characterizations build largely on established channel-synthesis techniques and provide a basis for analyzing the restrictions imposed by reusing pretrained components. Our main contribution lies in identifying how these restrictions determine the conditions and costs of transfer, thereby clarifying the distinct roles of the retained renderer, representation, and code.}

{
{\noindent \em Notation:} We closely follow the notation in~\cite{ElGamalKim2011}. Unless stated otherwise (e.g. Gaussian cases), all alphabets are finite. {Uppercase and lowercase letters denote random variables and realizations, respectively.} For $1\le i\le j$, write $X_i^j=(X_i,\ldots,X_j)$ and $X^n=X_1^n$. {Product targets are explicit in distributional criteria. Information quantities use the distribution specified in each statement or proof. Event probabilities and expectations use $\mathsf P$ and $\E$, respectively. Kernel compositions are explicit sums over intermediate variables.}

The notation $\mathbf 1\{\cdot\}$ denotes an indicator. For a pmf $p(a)$, define $\suppo(p(a)):=\{a:p(a)>0\}$.
For $a\in\Ac$, let $\delta_a$ denote the point-mass distribution at $a$, i.e., $\delta_a(a'):=\mathbf 1\{a'=a\}$.
For a finite set $\Ac$, $\Pc(\Ac)$ denotes the probability simplex on $\Ac$; $\conv(\Ac)$ and $\cl(\Ac)$ denote the convex hull and closure of a set $\Ac$, respectively.

For distributions $p$ and $r$ on the same finite alphabet,
\begin{align}
    \|p-r\|_{\TV}
    &:=\frac12\sum_a|p(a)-r(a)|.
    \label{eq:tv-definition}
\end{align}
A stochastic map contracts total variation. We repeatedly use this fact to pass from strong synthesis of one collection of random variables to a marginal or postprocessed collection.
}

{Section~\ref{sec:preliminaries} reviews strong-perception RDP coding and the full channel-synthesis benchmark. Section~\ref{sec:problem} states the problem, and Section~\ref{sec:main-results} presents the results for the three transfer strategies. Section~\ref{sec:examples} gives discrete separations and the scalar Gaussian specialization. Section~\ref{sec:discussion} discusses positioning, limitations, and extensions. Detailed proofs appear in the appendices.}

\section{Pretrained System and Transfer Task}
\label{sec:preliminaries}

We first review strong-perception RDP coding, which describes the pretrained system whose components are available for reuse. We then present channel synthesis, which provides the operational formulation of the new conditional-generation task to be performed using those components.

\subsection{Strong-Perception RDP Coding}
\label{subsec:strong-rdp}

\begingroup

Consider a discrete memoryless source (DMS) $X\sim q_0(x)$ with alphabet $\mathcal X$ and a bounded distortion measure $d:\mathcal X\times\mathcal X\to[0,\infty)$. For a source block $X^n\sim q_0(x^n)=\prod_{i=1}^n q_0(x_i)$, let $d_n(x^n,\tilde x^n):=\sum_{i=1}^n d(x_i,\tilde x_i)$.

A $(2^{nR_0},2^{nR_{c,0}},n)$ RDP code $\mathsf C_{0,n}$ with renderer $W_0$ consists of
\begin{itemize}

    \item a message set $[1:2^{nR_0}]$ and a common-randomness set $[1:2^{nR_{c,0}}]$,

    \item a stochastic encoder that assigns a distribution $p(m_0|x^n,k_0)$ on the message set to each pair $(x^n,k_0)\in\mathcal X^n\times[1:2^{nR_{c,0}}]$, and

    \item a decoder specified by a latent codeword $u^n(m_0,k_0)\in\mathcal U^n$ for each pair $(m_0,k_0)$, which generates the reconstruction according to $W_0(\tilde x^n|u^n(m_0,k_0))$, where $W_0(\tilde x^n|u^n):=\prod_{i=1}^n W_0(\tilde x_i|u_i)$.
    \end{itemize}
The encoder and decoder share $K_0\sim\U[1:2^{nR_{c,0}}]$, independent of $X^n$. The induced message, latent sequence, and reconstruction are denoted by $M_0$, $U^n=u^n(M_0,K_0)$, and $\tilde X^n$, respectively. Thus, $R_0$ and $R_{c,0}$ are the compression rate and common-randomness rate, respectively.

A sequence of such codes satisfies the strong near-perfect perception criterion if
\begin{align}
    \left\|p(\tilde x^n)-q_0(x^n)\right\|_{\TV}
    \longrightarrow0
    \label{eq:strong-near-perfect}
\end{align}
and the distortion constraint $D_0$ if
\begin{align}
    \limsup_{n\to\infty}
    \frac1n\E\!\left[d_n(X^n,\tilde X^n)\right]
    \le D_0.
    \label{eq:base-distortion-operational}
\end{align}

To describe the pretrained representation separately from its operational block code, we introduce the following single-letter factorization.

\begin{definition}[Single-Letter RDP Factorization]
\label{def:rdp-factorization}
For a DMS $X\sim q_0(x)$, a tuple
\begin{align}
    \baseG=(\Uc,q_0(u|x),W_0)
    \label{eq:base-model}
\end{align}
is called an $(R_0,R_{c,0},D_0)$ single-letter RDP factorization if, under
\begin{align}
    q_0(x,u,\tilde x)
    =q_0(x)q_0(u|x)W_0(\tilde x|u),
    \label{eq:base-factorization-law}
\end{align}
we have
\begin{align}
    q_0(\tilde x=x)&=q_0(x),
    \label{eq:base-perfect-single}\\
    \E[d(X,\tilde X)]&\le D_0,
    \label{eq:base-single-distortion}\\
    R_0&\ge I(X;U),
    \label{eq:base-rate-r}\\
    R_0+R_{c,0}&\ge I(\tilde X;U).
    \label{eq:base-rate-sum}
\end{align}
We take
\begin{align}
    \mathcal U=\operatorname{supp}(q_0(u)),
    \label{eq:effective-support}
\end{align}
so that every renderer input is used with positive probability under the factorization.
\end{definition}

The joint distribution in Definition~\ref{def:rdp-factorization} satisfies $X-U-\tilde X$. Here $q_0(u|x)$ specifies the single-letter representation statistics. We relate these statistics to an operational block code through a compatibility condition.

For a code $\mathsf C_{0,n}$, the actual joint distribution is
\begin{align}
&p(x^n,k_0,m_0,u^n,\tilde x^n)
\nonumber\\
&\quad=
q_0(x^n)\,2^{-nR_{c,0}}
p(m_0|x^n,k_0)
\mathbf 1\{u^n=u^n(m_0,k_0)\}
W_0(\tilde x^n|u^n).
\label{eq:actual-rdp-law}
\end{align}
Define the corresponding ideal distribution by
\begin{align}
&\bar p(x^n,k_0,m_0,u^n,\tilde x^n)
\nonumber\\
&\quad=
2^{-n(R_0+R_{c,0})}
\mathbf 1\{u^n=u^n(m_0,k_0)\}
q_0(x^n|u^n)
W_0(\tilde x^n|u^n),
\label{eq:ideal-rdp-law}
\end{align}
where $q_0(x|u)$ is the reverse conditional distribution associated with $q_0(x,u)$, and $q_0(x^n|u^n)=\prod_{i=1}^n q_0(x_i|u_i)$. Under this ideal distribution, $(M_0,K_0)$ is uniform, $U^n=u^n(M_0,K_0)$, and $X^n$ and $\tilde X^n$ are generated conditionally independently given $U^n$ according to $q_0(x^n|u^n)$ and $W_0(\tilde x^n|u^n)$, respectively.

\begin{definition}[Compatible RDP Code Sequence]
\label{def:compatible-rdp-code}
A sequence of operational RDP codes $\mathsf C_0=\{\mathsf C_{0,n}\}_{n\ge1}$ is compatible with the single-letter factorization $\baseG$ if
\begin{align}
    \left\|
    p(x^n,k_0,m_0,u^n,\tilde x^n)
    -\bar p(x^n,k_0,m_0,u^n,\tilde x^n)
    \right\|_{\TV}
    \longrightarrow0.
    \label{eq:rdp-compatibility}
\end{align}
\end{definition}

Compatibility requires the actual joint distribution to approach the ideal distribution associated with the factorization. It does not require the operational encoder to satisfy $p(u^n|x^n)=q_0(u^n|x^n)$. When $R_0>I(X;U)$ and $R_0+R_{c,0}>I(\tilde X;U)$, the likelihood-encoder construction in~{\cite[proof of Thm.~2, eq.~(82)]{Wagner2022}} yields a compatible code sequence satisfying the strong perception and distortion constraints above.

Finally, the single-letter factorization induces the end-to-end RDP test channel
\begin{align}
    T_0(\tilde x|x)
    :=\sum_{u\in\mathcal U}q_0(u|x)W_0(\tilde x|u).
    \label{eq:t0-definition}
\end{align}
For a pmf $p$ and a channel $T$, write $(pT)(x'):=\sum_xp(x)T(x'|x)$. The perfect-perception condition in Definition~\ref{def:rdp-factorization} is then equivalent to
\begin{align}
    q_0T_0=q_0.
    \label{eq:px-stationary}
\end{align}
Neither the single-letter factorization $\baseG$ nor a compatible operational implementation $\mathsf C_0$ is assumed to be unique.
\endgroup

\subsection{Channel Synthesis and the Full Redesign Benchmark}
\label{subsec:full-synthesis}

\begingroup

We now consider conditional synthesis, where a semantic source $S$ specifies the desired output distribution $q(x|s)$. Designing both the encoder and decoder freely for this task gives the Full Redesign benchmark. We follow the distributed channel-synthesis formulation in~\cite{Cuff2013}.

Let
\begin{align}
    q(s,x)=q(s)q(x|s)
    \label{eq:semantic-target-prelim}
\end{align}
be a prescribed finite-alphabet transfer target distribution, and write $q(s^n,x^n):=\prod_{i=1}^n q(s_i,x_i)$, with associated marginals and conditional distributions denoted analogously. The encoder observes $S^n\sim q(s^n)$, and the decoder generates $\hat X^n$ so that $(S^n,\hat X^n)$ approaches the prescribed joint distribution.

A $(2^{nR},2^{nR_c},n)$ synthesis code consists of
\begin{itemize}

    \item a message set $[1:2^{nR}]$ and a common-randomness set $[1:2^{nR_c}]$,

    \item a stochastic encoder that assigns a distribution $p(m|s^n,k)$ on the message set to each pair $(s^n,k)\in\mathcal S^n\times[1:2^{nR_c}]$, and

    \item a stochastic decoder that assigns a distribution $p(\hat x^n|m,k)$ on $\mathcal X^n$ to each pair $(m,k)\in{[1:2^{nR}]}\times{[1:2^{nR_c}]}$.
\end{itemize}
The encoder and decoder share $K\sim\U[1:2^{nR_c}]$, independent of $S^n$. The compression rate and common-randomness rate are $R$ and $R_c$, respectively. The message and reconstruction are generated according to $p(m|s^n,k)$ and $p(\hat x^n|m,k)$, respectively.

A rate pair $(R,R_c)$ is achievable if there exists a sequence of such codes satisfying the strong synthesis criterion
\begin{align}
    \left\|p(s^n,\hat x^n)-\prod_{i=1}^n q(s_i,\hat x_i)\right\|_{\TV}\to0.
    \label{eq:full-synthesis-criterion}
\end{align}
The Full Redesign region $\mathscr R_{\full}(q)$ is the closure of the set of achievable rate pairs. No component of the pretrained RDP system is {transferred} (retained).
\endgroup

\begin{theorem}[Full Redesign Region~\cite{Cuff2013}]
\label{thm:cuff-full}
Let
\begin{align}
    \mathcal D(q)
    &:=
    \left\{
    p(s,v,x):
    \begin{array}{l}
        p(s,x)=q(s,x),\quad S-V-X,\\
        |\Vc|\le |\Sc||\Xc|+1
    \end{array}
    \right\}.
    \label{eq:full-admissible-laws}
\end{align}
The Full Redesign region is
\begingroup
\begin{align}
    \mathscr R_{\full}(q)
    &=\bigcup_{p\in\mathcal D(q)}
    \left\{
    (R,R_c):
    \begin{aligned}
        R     &\ge I(S;V),\\
        R+R_c &\ge I(S,X;V)
    \end{aligned}
    \right\},
    \label{eq:full-region}
\end{align}
\endgroup
where $p(s,v,x)$ is a joint pmf.
\end{theorem}

For fixed $R_c$, define
\begin{align}
    \Rfull(R_c;q)
    &:=
    \min_{p\in\mathcal D(q)}
    \max\bigl\{I(S;V),I(S,X;V)-R_c\bigr\}.
    \label{eq:full-boundary}
\end{align}
The two classical endpoints are
\begin{align}
    \Rfull(0;q)&=C_{\rm W}(S;X),
    \label{eq:full-wyner}\\
    \lim_{R_c\to\infty}\Rfull(R_c;q)&=I(S;X),
    \label{eq:full-mutual}
\end{align}
where $C_{\rm W}$ is Wyner's common information~\cite{Wyner1975}.

\section{Problem Statement}
\label{sec:problem}


We consider a future transfer task for a pretrained RDP system. The semantic encoder observes $S^n\sim q(s^n)$ and communicates with a decoder that generates $\hat X^n$. Their objective is to approximate the prescribed joint law $\prod_{i=1}^n q(s_i,\hat x_i)$, while retaining specified components of the pretrained system.

The pretrained system is described by an RDP factorization
\begin{align}
    \baseG=(\Uc,q_0(u|x),W_0)
    \label{eq:fixed-G0}
\end{align}
in the sense of Definition~\ref{def:rdp-factorization}, {\em fixed before the task is specified}. We consider three forms of reuse: retaining only the renderer $W_0$, additionally retaining the semantic latent law induced by $q_0(u|x)$, or retaining an entire compatible RDP code sequence
\begin{align}
    \mathsf C_0=\{\mathsf C_{0,n}\}_{n\ge1}
    \label{eq:fixed-C0}
\end{align}
as defined in Definition~\ref{def:compatible-rdp-code}. These choices determine which components may be redesigned for the new task.

Renderer-Only Transfer allows an arbitrary transfer target marginal $q(x)$. For the main Representation-Induced and Full-Model Transfer results, we assume
\begin{align}
    q(x)=q_0(x).
    \label{eq:matched-marginal-main}
\end{align}
The transfer target therefore has the same output marginal as the base target distribution $q_0(x)$, while the conditional law $q(x|s)$ specifies the desired output distribution for each semantic condition $s$.

\subsection{Renderer-Only Transfer}
\label{subsec:renderer-only-definition}

{Universal RDP and rate--distortion--classification representations retain an encoder representation and adapt the decoder across objectives~\cite{ZhangEtAl2025,NguyenEtAl2025}. In contrast, Renderer-Only Transfer retains the generative renderer and redesigns the condition-to-latent steering.}

\begingroup

A $(2^{nR},2^{nR_c},n)$ Renderer-Only Transfer code consists of
\begin{itemize}

    \item a message set $[1:2^{nR}]$ and a common-randomness set $[1:2^{nR_c}]$,

    \item a stochastic encoder that generates a message $M\in[1:2^{nR}]$ from each pair $(s^n,k)\in\mathcal S^n\times[1:2^{nR_c}]$, and

    \item a stochastic latent-generation decoder that generates $U^n\in\mathcal U^n$ from each pair $(m,k)\in{[1:2^{nR}]}\times{[1:2^{nR_c}]}$.
\end{itemize}
Denote the encoder and latent-decoder distributions by
\begin{align}
    p(m|s^n,k),
    \qquad
    p(u^n|m,k),
    \label{eq:reuse-code-kernels}
\end{align}
respectively. The encoder and decoder share $K\sim\U[1:2^{nR_c}]$, independent of $S^n$. Given $U^n=u^n$, the inherited renderer generates $\hat X^n$ according to
\begin{align}
    W_0(\hat x^n|u^n)
    =\prod_{i=1}^nW_0(\hat x_i|u_i).
    \label{eq:reuse-renderer-product}
\end{align}
Thus the induced joint distribution is
\begin{align}
&p(s^n,k,m,u^n,\hat x^n)\nonumber\\
&\quad=q(s^n)\,2^{-nR_c}p(m|s^n,k)p(u^n|m,k)W_0(\hat x^n|u^n).
\label{eq:reuse-induced-law}
\end{align}
No prescribed distribution is imposed on $U^n$. The encoder and latent-generation decoder may be designed to meet the output synthesis criterion through the fixed renderer $W_0$.

\begin{definition}[Renderer-Only Transfer Region]
\label{def:reuse-region}
A pair $(R,R_c)$ is achievable for $(q(s,x),W_0)$ under Renderer-Only Transfer if there exists a sequence of codes of the form~\eqref{eq:reuse-code-kernels}--\eqref{eq:reuse-renderer-product} satisfying
\begin{align}
    \left\|p(s^n,\hat x^n)-\prod_{i=1}^n q(s_i,\hat x_i)\right\|_{\TV}\to0.
    \label{eq:reuse-criterion}
\end{align}
The Renderer-Only Transfer region $\mathscr R_{\ro}(q|W_0)$ is the closure of the set of achievable rate pairs.
\end{definition}
\endgroup

For a given common-randomness rate $R_c$, define the minimum required {compression rate} as
\begin{align}
    \Rro(R_c;q,W_0)
    :=\inf\{R:(R,R_c)\in\mathscr R_{\ro}(q|W_0)\}.
    \label{eq:reuse-boundary-definition}
\end{align}

\subsection{Representation-Induced Transfer}
\label{subsec:ri-definition}
The Renderer-Only Transfer model does not prescribe the latent distribution. We now consider a more restrictive form of transfer in which the pretrained representation determines the target conditional distribution of the latent variable given the semantic source. For a matched-marginal task~\eqref{eq:matched-marginal-main}, define
\begin{align}
    B_0(u|s)
    &:=\sum_xq(x|s)q_0(u|x)
    \label{eq:ri-latent-law}
\end{align}
and the corresponding single-letter output kernel
\begin{align}
(B_0W_0)(\hat x|s)
&=\sum_uB_0(u|s)W_0(\hat x|u) =\sum_xq(x|s)T_0(\hat x|x).
\label{eq:ri-output-law}
\end{align}
{Products of named kernels, such as $B_0W_0$ in~\eqref{eq:ri-output-law}, denote channel composition.}
In a learned latent-variable model, $B_0$ is the semantic or class-conditional aggregate latent distribution induced by the existing representation. Representation-Induced Transfer therefore models using these inherited latent statistics rather than learning a new semantic steering law for the fixed renderer. The kernel $B_0$ is a latent target law induced by the transferred representation, not a claim that the original RDP encoder operates symbolwise according to $q_0(u|x)$.

{A $(2^{nR},2^{nR_c},n)$ Representation-Induced Transfer code uses the encoder, latent-generation decoder, and inherited renderer specified in Section~\ref{subsec:renderer-only-definition}, with the induced joint distribution~\eqref{eq:reuse-induced-law}. The semantic encoder and latent-generation decoder are newly designed, while $B_0$ and $W_0$ remain fixed. A sequence of such codes is additionally required to synthesize the inherited semantic latent distribution:}
\begin{align}
    \left\|p(s^n,u^n)-\prod_{i=1}^n q(s_i)B_0(u_i|s_i)\right\|_{\TV}
    &\to0.
    \label{eq:ri-latent-criterion}
\end{align}
By contraction of total variation through $W_0$,
\begin{align}
    \left\|p(s^n,\hat x^n)-\prod_{i=1}^n q(s_i)(B_0W_0)(\hat x_i|s_i)\right\|_{\TV}
    &\to0.
    \label{eq:ri-induced-output}
\end{align}
Thus Representation-Induced Transfer realizes the prescribed transfer target distribution exactly only when
\begin{align}
    (B_0W_0)(x|s)
    &=q(x|s),\qquad q(s)>0.
    \label{eq:ri-compatibility}
\end{align}
We say that $q(s,x)$ is \emph{representation-induced compatible with $\baseG$} if~\eqref{eq:ri-compatibility} holds.

\begin{definition}[Representation-Induced Transfer Region]
\label{def:ri-region}
{If~\eqref{eq:ri-compatibility} holds, a pair $(R,R_c)$ is achievable for $(q(s,x),\baseG)$ under Representation-Induced Transfer if there exists a sequence of such codes satisfying~\eqref{eq:ri-latent-criterion}. The Representation-Induced Transfer region $\mathscr R_{\ri}(q|\baseG)$ is the closure of the set of achievable rate pairs. If~\eqref{eq:ri-compatibility} fails, we set $\mathscr R_{\ri}(q|\baseG)=\emptyset$.}
\end{definition}

For a given common-randomness rate $R_c$, define the minimum required {compression rate} under Representation-Induced Transfer as
\begin{align}
    \Rri(R_c;q,\baseG)
    &:=\inf\{R:(R,R_c)\in\mathscr R_{\ri}(q|\baseG)\}.
    \label{eq:ri-boundary-definition}
\end{align}

Whenever~\eqref{eq:ri-compatibility} holds, every Representation-Induced Transfer code is also a Renderer-Only Transfer code, and hence
\begin{align}
    \mathscr R_{\ri}(q|\baseG)
    \subseteq
    \mathscr R_{\ro}(q|W_0)
    \subseteq
    \mathscr R_{\full}(q).
    \label{eq:ri-region-inclusion-problem}
\end{align}

\subsection{Full-Model Transfer}
\label{subsec:full-model-definition}

\begingroup

Representation-Induced Transfer retains the single-letter statistics of the pretrained representation while designing a new semantic synthesis code. Full-Model Transfer instead reuses a fixed compatible RDP code sequence $\mathsf C_0$ from~\eqref{eq:fixed-C0}, including its encoder, latent codebook, common-randomness mechanism, and renderer.

For a matched-marginal task, the semantic encoder observes $S^n\sim q(s^n)$ and shares the original seed $K_0\sim\U[1:2^{nR_{c,0}}]$ with the decoder, independently of $S^n$. It generates $M_0\in[1:2^{nR_0}]$ according to
\begin{align}
    p(m_0|s^n,k_0)
    &=\sum_{x^n}q(x^n|s^n)p(m_0|x^n,k_0),
    \label{eq:fm-encoder}
\end{align}
where $p(m_0|x^n,k_0)$ is the fixed RDP encoder. Given $(M_0,K_0)$, the decoder selects $U^n=u^n(M_0,K_0)$ from the original codebook and generates $\tilde X^n$ according to $W_0(\tilde x^n|U^n)$.

Thus the semantic encoder is determined by the task distribution and the original code. The compression and common-randomness rates remain $(R_0,R_{c,0})$, as specified in Section~\ref{subsec:strong-rdp}, and are not optimized for the transfer task. We retain $M_0,K_0$ for the inherited code, in contrast to $M,K$ for newly designed semantic codes. The task distribution $q(x|s)$ is treated as known, and its learning cost is outside the model.

For analysis, use the equivalent representation with $(S^n,X^n)\sim q(s^n,x^n)$ independent of $K_0$, followed by the original RDP code. Matched marginals give $X^n\sim q_0(x^n)$ in this representation. Marginalizing $X^n$ from the following joint distribution recovers the actual distribution induced by the semantic encoder and decoder:
\begin{align}
&p(s^n,x^n,k_0,m_0,u^n,\tilde x^n)
\nonumber\\
&\quad=
q(s^n,x^n)\,2^{-nR_{c,0}}
p(m_0|x^n,k_0)
\mathbf 1\{u^n=u^n(m_0,k_0)\}
W_0(\tilde x^n|u^n).
\label{eq:fm-induced-law}
\end{align}

\begin{definition}[Full-Model Transferability]
\label{def:full-model-transfer}
A compatible RDP code sequence $\mathsf C_0$ is said to be \emph{full-model transferable} for the matched-marginal transfer task $q(s,x)$ if
\begin{align}
    \left\|
    p(s^n,\tilde x^n)
    -\prod_{i=1}^n q(s_i,\tilde x_i)
    \right\|_{\TV}
    &\longrightarrow0,
    \label{eq:fm-criterion}
\end{align}
where $p(s^n,\tilde x^n)$ is the marginal of~\eqref{eq:fm-induced-law}.
\end{definition}
\endgroup

\begin{remark}[RDP Factorization Versus Full Test-Channel Synthesis]
The single-letter test channel $T_0$ describes the pretrained RDP factorization, but the actual pretrained block code is not required to implement $T_0$ independently across coordinates. The strong RDP criterion requires the output distribution to approach $q_0(x^n)$ and constrains the average distortion. It does not require
\begin{align}
    p(x^n,\tilde x^n)
    \approx
    \prod_{i=1}^n q_0(x_i)T_0(\tilde x_i|x_i).
\end{align}
If this stronger requirement held, the induced semantic output distribution would approach
\begin{align}
    \prod_{i=1}^n q(s_i)
    \sum_{x_i}q(x_i|s_i)T_0(\tilde x_i|x_i).
\end{align}
The single-letter condition $\sum_x q(x|s)T_0(\tilde x|x)=q(\tilde x|s)$ would then be necessary for Full-Model Transfer. Without this stronger requirement, however, failure of the single-letter condition does not by itself establish that the actual block code fails to transfer the task.

Theorem~\ref{thm:operational-transfer} shows that this single-letter condition is sufficient for every compatible pretrained block code and gives a converse bound on the semantic dependence that any transferable task can retain.
\end{remark}

Representation-Induced Transfer provides the nominal single-letter prediction associated with $\baseG$. Section~\ref{sec:full-model} studies the actual frozen block implementation, {which is guaranteed to preserve the nominal recurrent-class invariant but may possess additional transfer capability depending on its {compression rate} and operational structure.}

The four systems can therefore be summarized as follows:
\begin{center}
\begin{tabular}{@{}lll@{}}
\toprule
System & Fixed Component(s) & Redesigned for the Task \\
\midrule
Full Redesign & None & Synthesis Auxiliary and Generator \\
Renderer-Only Transfer & $W_0$ & Semantic Synthesis and Latent Steering \\
Representation-Induced Transfer & $q_0(u|x)$ and $W_0$ & Semantic Code for the Induced Latent Law \\
Full-Model Transfer & $\mathsf C_0$ & {None (Induced Encoder~\eqref{eq:fm-encoder})} \\
\bottomrule
\end{tabular}
\end{center}

\section{Main Results}
\label{sec:main-results}

The results are stated in order of increasing transfer restriction.  We first ask what the inherited renderer can accomplish when its semantic steering is redesigned under Renderer-Only Transfer.  We then retain the representation-induced latent statistics under Representation-Induced Transfer and identify the tasks that already transfer without task-specific latent adaptation.  Finally, we retain an actual operational RDP code, identify a task-independent transferable family, and bound any additional semantic dependence that a particular implementation can preserve.

\subsection{Renderer-Only Transfer}
\label{sec:adaptive}

We begin with the least restrictive transfer architecture. Only $W_0$ is retained from the base RDP system; the semantic encoder, synthesis auxiliary, and condition-to-latent mapping may all be redesigned. Comparing Renderer-Only Transfer with Full Redesign isolates the cost of retaining the generator itself before asking whether the original representation law should also be transferred.

\subsubsection{Exact Transfer Region}

The first theorem keeps the usual channel-synthesis architecture intact and changes only the set of decoder channels that are admissible. To make this restriction explicit, define the renderer output hull
\begin{align}
    \Cren
    &:=\conv\{W_0(x|u):u\in\Uc\}
    \subseteq\Pc(\Xc).
    \label{eq:renderer-hull}
\end{align}
The following reduction will be used repeatedly.

\begin{lemma}[Decoder-Side Reduction and Renderer-Hull Characterization]
\label{lem:decoder-side-reduction}
Fix a renderer $W_0$. For any channel $p(u|v)$, define
\begin{align}
    \overline W_0(x|v)
    &:=\sum_up(u|v)W_0(x|u).
    \label{eq:effective-renderer}
\end{align}
Marginalizing $U$ from $p(s,v,u,x)=q(s)p(v|s)p(u|v)W_0(x|u)$ gives $p(s,v,x)=q(s)p(v|s)\overline W_0(x|v)$. Moreover, a channel $p(x|v)$ admits a factorization
\begin{align}
    p(x|v)
    &=\sum_up(u|v)W_0(x|u)
    \label{eq:decoder-factorization}
\end{align}
for some $p(u|v)$ if and only if
\begin{align}
    p(x|v)
    &\in\Cren,
    \qquad v\in\suppo(p(v)).
    \label{eq:decoder-row-hull}
\end{align}
\end{lemma}

\begin{proof}
The first statement is immediate from marginalizing $U$. Hence both factorizations induce the same $p(s,v,x)$ and, in particular, the same $I(S;V)$ and $I(S,X;V)$.

For the second, if~\eqref{eq:decoder-factorization} holds, then each $p(x|v)$ is a convex combination of the rows of $W_0$ and therefore belongs to $\Cren$. For the reverse implication, suppose~\eqref{eq:decoder-row-hull} holds. By the definition of the convex hull, for every $v\in\suppo(p(v))$ there exist coefficients $\lambda(u|v)\ge0$ with $\sum_u\lambda(u|v)=1$ such that $p(x|v) =\sum_u\lambda(u|v)W_0(x|u)$. Defining $p(u|v):=\lambda(u|v)$ gives~\eqref{eq:decoder-factorization}; its values outside $\suppo(p(v))$ may be chosen arbitrarily.
\end{proof}

Define
\begin{align}
    \mathscr P_{W_0}(q)
    &:=\left\{
    p(s,v,u,x):
    \begin{array}{l}
        p(s,x)=q(s,x),\\
        S-V-U-X,\quad p(x|u)=W_0(x|u)
    \end{array}
    \right\}.
    \label{eq:reuse-family}
\end{align}
Here $p(s,v,u,x)$ is a joint pmf.

\begin{theorem}[Renderer-Only Transfer Region]
\label{thm:reuse-region}
For finite alphabets,
\begingroup
\begin{align}
    \mathscr R_{\ro}(q|W_0)
    &=\bigcup_{p\in\mathscr P_{W_0}(q)}
    \left\{
    (R,R_c):
    \begin{aligned}
        R     &\ge I(S;V),\\
        R+R_c &\ge I(S,X;V)
    \end{aligned}
    \right\}.
    \label{eq:reuse-region}
\end{align}
\endgroup
It suffices to take $|\Vc|\le |\Sc||\Xc|+1$.
\end{theorem}

The proof is given in Appendix~\ref{app:reuse-region}.
Theorem~\ref{thm:reuse-region} shows that retaining $W_0$ introduces no new type of {coding resource}: it restricts only which decoder channels may be used in Cuff's synthesis region. Equivalently, every conditional decoder row selected by the synthesis auxiliary must lie in the renderer hull $\Cren$. This observation separates two distinct limitations of transfer: a target may lie outside the renderer's range altogether, or it may be realizable only through a less rate-efficient factorization.

\subsubsection{Semantic Controllability}

Before comparing rates, it is useful to separate a more basic question: can the fixed renderer generate the requested conditional laws at any rate? The answer is row-wise and geometric.

\begin{corollary}[Exact Controllability]
\label{cor:controllability}
The following are equivalent:
\begin{enumerate}[label=(\roman*)]

    \item $\mathscr R_{\ro}(q|W_0)\ne\emptyset$;

    \item there exists $p(u|s)$ such that
    \begin{align}
        q(x|s)
        &=\sum_up(u|s)W_0(x|u);
        \label{eq:reuse-factorization-feasible}
    \end{align}

    \item for every $s\in\suppo(q(s))$,
    \begin{align}
        q(x|s)
        &\in\Cren.
        \label{eq:row-hull-condition}
    \end{align}
\end{enumerate}
\end{corollary}
\begin{proof}
If $\mathscr R_{\ro}(q|W_0)\ne\emptyset$, Theorem~\ref{thm:reuse-region} gives a joint distribution $p(s,v,u,x)=q(s)p(v|s)p(u|v)W_0(x|u)$ with $(S,X)$ marginal $q(s,x)$. Marginalizing $V$ yields $p(u|s) =\sum_vp(v|s)p(u|v)$, and hence $q(x|s)=\sum_u p(u|s)W_0(x|u)$. If such a $p(u|s)$ exists, choosing $V=U$ in Theorem~\ref{thm:reuse-region} shows that the Renderer-Only Transfer region is nonempty. Finally, $(ii)\Longleftrightarrow(iii)$ is immediate from the definition of $\Cren$.
\end{proof}
We say that a transfer target distribution $q(s,x)$ is \emph{controllable through $W_0$} if it satisfies any of the equivalent conditions in Corollary~\ref{cor:controllability}. Operationally, condition (ii) provides a single-letter criterion for determining whether the prescribed transfer target distribution is achievable through the fixed renderer. Moreover, since the base RDP factorization satisfies $q_0(x) =\sum_uq_0(u)W_0(x|u)$, we always have $q_0(x)\in\Cren$. This guarantees that the renderer can reproduce the unconditional base law, but it does not ensure that it can reproduce each conditional law $q(x|s)$ whose mixture is $q_0(x)$. Thus strong unconditional realism and conditional controllability are distinct requirements.

\subsubsection{Price of Transfer Relative to Full Redesign}

When the target is controllable, the next question is how much performance is lost by retaining $W_0$. By Lemma~\ref{lem:decoder-side-reduction}, every Renderer-Only Transfer factorization also defines a valid Full Redesign factorization, and hence
\begin{align}
    \mathscr R_{\ro}(q|W_0)
    &\subseteq\mathscr R_{\full}(q).
    \label{eq:reuse-in-full}
\end{align}
Define the fixed-$R_c$ price of Renderer-Only Transfer by
\begin{align}
    \Deltaro(R_c;q)
    &:=\Rro(R_c;q,W_0)-\Rfull(R_c;q)
    \in[0,\infty].
    \label{eq:finite-price-reuse}
\end{align}
The value $+\infty$ represents lack of controllability; a finite positive value represents a {compression-rate penalty} despite exact feasibility.

\begin{proposition}[Zero Finite-$R_c$ Transfer Price]
\label{prop:finite-zero-price}
For finite alphabets, $\Deltaro(R_c;q)=0$ if and only if at least one optimizer of the Full Redesign boundary~\eqref{eq:full-boundary} can be chosen so that its conditional output distributions satisfy
\begin{align}
    p(x|v)
    &\in\Cren,
    \qquad v\in\suppo(p(v)).
    \label{eq:optimizer-in-hull}
\end{align}
\end{proposition}

\begin{proof}
If a Full Redesign optimizer satisfies $p(x|v)\in\Cren$ for every $v\in\suppo(p(v))$, then Lemma~\ref{lem:decoder-side-reduction} factors its decoder channel through $W_0$ without changing $I(S;V)$ or $I(S,X;V)$, and hence the same boundary point is achievable under Renderer-Only Transfer.

For the reverse implication, suppose $\Deltaro(R_c;q)=0$. For finite alphabets, choose an optimizer of the Renderer-Only Transfer boundary. Its factorization $S-V-U-X$ satisfies
\begin{align}
    p(x|v)
    &=\sum_up(u|v)W_0(x|u)
    \in\Cren,
    \qquad v\in\suppo(p(v)).
\end{align}
After marginalizing $U$, the resulting $S-V-X$ factorization is feasible for Full Redesign and achieves the same boundary value. Since the two boundary values are equal, this factorization is also an optimizer of the Full Redesign problem.
\end{proof}

Proposition~\ref{prop:finite-zero-price} shows that zero transfer price at a given $R_c$ does not require $W_0$ to implement every possible decoder channel. It is enough that one Full Redesign optimizer at that operating point uses only decoder distributions already contained in $\Cren$.

\subsubsection{Unlimited Common Randomness}

\begin{theorem}[Price of Renderer-Only Transfer With Unlimited Common Randomness]
\label{thm:unlimited-reuse}
If $q(s,x)$ is controllable through $W_0$, then
\begin{align}
    R_{\ro}^{\infty}(q|W_0)
    &:=\lim_{R_c\to\infty}\Rro(R_c;q,W_0)
    \nonumber\\
    &=\min_{p(u|s):\,\sum_u p(u|s)W_0(x|u)=q(x|s)}I(S;U).
    \label{eq:ro-unlimited-rate}
\end{align}
Moreover, relative to the Full Redesign benchmark
\begin{align}
    R_{\full}^{\infty}(q)
    &:=\lim_{R_c\to\infty}\Rfull(R_c;q) =I(S;X),
\end{align}
{the compression-rate penalty} of Renderer-Only Transfer is
\begin{align}
    \Gammaro(q)
    &:=R_{\ro}^{\infty}(q|W_0)-R_{\full}^{\infty}(q)
    \nonumber\\
    &=\min_{p(u|s):\,\sum_u p(u|s)W_0(x|u)=q(x|s)}I(S;U|X).
    \label{eq:gamma-definition}
\end{align}
If the target is not controllable through $W_0$, we set $R_{\ro}^{\infty}(q|W_0)=\Gammaro(q)=+\infty$.
\end{theorem}

\begin{proof}
Under unlimited common randomness, the sum-rate constraint in Theorem~\ref{thm:reuse-region} becomes inactive, and therefore
\begin{align}
    R_{\ro}^{\infty}(q|W_0)
    &=\min_{p\in\mathscr P_{W_0}(q)}I(S;V).
    \label{eq:unlimited-ro-v-form}
\end{align}
We first lower-bound this expression. Any feasible factorization $S-V-U-X$ induces a channel $p(u|s)$ satisfying $q(x|s)=\sum_u p(u|s)W_0(x|u)$. Moreover, $S-V-U$ and the data-processing inequality give $I(S;V) \ge I(S;U)$. Hence
\begin{align}
    R_{\ro}^{\infty}(q|W_0)
    &\ge\min_{p(u|s):\,\sum_u p(u|s)W_0(x|u)=q(x|s)}I(S;U).
\end{align}
For the reverse inequality, take any feasible $p(u|s)$ and choose $V=U$ in Theorem~\ref{thm:reuse-region}. This gives
\begin{align}
    R_{\ro}^{\infty}(q|W_0)
    &\le\min_{p(u|s):\,\sum_u p(u|s)W_0(x|u)=q(x|s)}I(S;U).
\end{align}
Combining the two inequalities proves~\eqref{eq:ro-unlimited-rate}.

We next compare Renderer-Only Transfer with Full Redesign. For every feasible $p(u|s)$, the induced law satisfies $S-U-X$ and has $(S,X)$-marginal $q(s,x)$. Therefore,
\begin{align}
    I(S;U)
    &=I(S;X)+I(S;U|X)-I(S;X|U)
    \nonumber\\
    &=I(S;X)+I(S;U|X),
\end{align}
where the second equality follows from $S-U-X$. Since $I(S;X)$ is fixed by the transfer target distribution,
\begin{align}
    R_{\ro}^{\infty}(q|W_0)
    &=I(S;X)
    +\min_{p(u|s):\,\sum_u p(u|s)W_0(x|u)=q(x|s)}I(S;U|X).
\end{align}
Subtracting $R_{\full}^{\infty}(q)=I(S;X)$ proves~\eqref{eq:gamma-definition}.
\end{proof}

Theorem~\ref{thm:unlimited-reuse} identifies the price of retaining the renderer with the least amount of semantic information that its input $U$ must carry beyond what is visible in the generated output $X$. Full Redesign needs only the dependence already present in $(S,X)$, whereas the inherited renderer may force the synthesis code to distinguish latent states that become indistinguishable after rendering.

\begin{corollary}[Task-Specific Zero-Price Renderer-Only Transfer]
\label{cor:task-zero-price}
For a controllable target, $\Gammaro(q)=0$ if and only if there exists a feasible latent $U$ satisfying both
\begin{align}
    S-U-X
    \qquad\text{and}\qquad
    S-X-U.
    \label{eq:double-markov-zero-price}
\end{align}
\end{corollary}

The proof is immediate from the feasibility condition $S-U-X$ and the equivalence $I(S;U|X)=0\Longleftrightarrow S-X-U$. Thus zero Renderer-Only Transfer price occurs when it is sufficient for $U$ not to carry extra semantic information about $S$ beyond what is already visible in $X$.

\subsubsection{Universal No-Loss Renderers}

A matched-marginal transfer task is any $q(s,x)$ with $q(x)=q_0(x)$. We say that $W_0$ is \emph{universally controllable over matched-marginal transfer targets} if every target in this class is controllable through $W_0$. {For the base source, write $\Xsupp:=\suppo(q_0(x))$.} The next theorem characterizes this property.

\begin{theorem}[Universal Matched-Marginal No-Loss Renderer-Only Transfer]
\label{thm:universal-no-loss}
The following are equivalent:
\begin{enumerate}[label=(\roman*)]

    \item every matched-marginal transfer target is controllable through $W_0$, in the sense that, for every finite semantic alphabet $\Sc$ and every $q(s,x)$ satisfying $q(x)=q_0(x)$, there exists a channel $p(u|s)$ such that $q(x|s) =\sum_up(u|s)W_0(x|u)$;
    \item $\Pc(\Xsupp) \subseteq\Cren$;
    \item for every $x\in\Xsupp$, there exists $u_x\in\Uc$ such that
    \begin{align}
        W_0(x'|u_x)
        &=\delta_x(x');
        \label{eq:point-mass-row}
    \end{align}

    \item for every matched-marginal transfer target $q$,
    \begin{align}
        \mathscr R_{\ro}(q|W_0)
        &=\mathscr R_{\full}(q).
        \label{eq:universal-region-equality}
    \end{align}
\end{enumerate}
\end{theorem}

The proof is given in Appendix~\ref{app:universal-no-loss}.

Theorem~\ref{thm:universal-no-loss} shows that universal controllability is substantially stronger than the perception requirement imposed on the base RDP system. The latter guarantees only that $q_0(x)$ belongs to the renderer hull $\Cren$. Universal controllability requires the renderer to realize every possible decomposition of $q_0(x)$ into conditional laws indexed by a semantic variable. Equivalently, $\Cren$ must contain the entire probability simplex on $\Xsupp$, including every point mass $\delta_x$. Thus a renderer may be fully adequate for unconditional perceptual reconstruction while being fundamentally incapable of supporting arbitrary conditional-generation tasks with the same output marginal.

\subsubsection{Renderer Ordering}

For two renderers with the same output alphabet, write $W_1\preceqin W_2$ if there exists a stochastic kernel $p(u_2|u_1)$ such that
\begin{align}
    W_1(x|u_1)
    &=\sum_{u_2}p(u_2|u_1)W_2(x|u_2).
    \label{eq:input-degradation}
\end{align}
This is input degradation, or input simulation, in the sense of~\cite{Nasser2018}.

\begin{proposition}[Monotonicity Under Input Simulation]
\label{prop:renderer-order}
If $W_1\preceqin W_2$, then for every transfer target $q$,
\begin{align}
    \mathscr R_{\ro}(q|W_1)
    &\subseteq\mathscr R_{\ro}(q|W_2),
    \label{eq:region-ordering}
\end{align}
and consequently $\Gamma_{W_2}(q) \le\Gamma_{W_1}(q)$. For finite alphabets, $W_1\preceqin W_2$ if and only if $\mathcal C(W_1) \subseteq\mathcal C(W_2)$. \end{proposition}

\begin{proof}
Take any Renderer-Only Transfer code for $W_1$. By~\eqref{eq:input-degradation}, the decoder can replace each use of $W_1$ by first drawing $U_2$ according to $p(u_2|u_1)$ and then applying $W_2$. The induced law of $(S^n,\hat X^n)$ and both rates are unchanged, proving~\eqref{eq:region-ordering}. The inequality between the unlimited-common-randomness prices follows because the Full Redesign benchmark is the same for both renderers.

The equivalence with hull inclusion is the renderer-hull characterization in Lemma~\ref{lem:decoder-side-reduction}: $W_1\preceqin W_2$ exactly when every row of $W_1$ is a convex combination of rows of $W_2$.
\end{proof}

Proposition~\ref{prop:renderer-order} formalizes the intuition that a renderer with a larger output hull cannot be worse for Renderer-Only Transfer. In particular, two finite-alphabet renderers with the same output hull have identical Renderer-Only Transfer regions for every transfer target, even if their latent parameterizations differ.

\subsection{Representation-Induced Transfer}
\label{sec:ri}

Renderer-Only Transfer describes what the renderer could generate after a new steering law is learned. The representation-induced scheme instead asks which semantic latent law is already induced by the pretrained statistics $q_0(u|x)$. It retains this law together with $W_0$, rather than optimizing a new steering channel. As defined in Section~\ref{subsec:ri-definition}, the semantic system does not execute $q_0(u|x)$ symbol by symbol. Instead, it designs a new synthesis code for the latent target distribution $p(s,u)=q(s)B_0(u|s)$, with $B_0(u|s)$ defined in~\eqref{eq:ri-latent-law}.

\subsubsection{Operational Region}

\begin{theorem}[Representation-Induced Transfer Region]
\label{thm:ri-region}
Suppose $q(s,x)$ is representation-induced compatible with $\baseG$, i.e., $(B_0W_0)(x|s)=q(x|s)$. Then
\begingroup
\begin{align}
    \mathscr R_{\ri}(q|\baseG)
    &=\bigcup_{\substack{p(s,v,u):\\
        p(s,u)=q(s)B_0(u|s),\;S-V-U}}
    \left\{
    (R,R_c):
    \begin{aligned}
        R     &\ge I(S;V),\\
        R+R_c &\ge I(S,U;V)
    \end{aligned}
    \right\}.
    \label{eq:ri-region-main}
\end{align}
\endgroup
It suffices to take $|\Vc|\le |\Sc||\Uc|+1$.
\end{theorem}

The proof follows that of Theorem~\ref{thm:reuse-region} in Appendix~\ref{app:reuse-region} with minor modifications, so we provide only a sketch.

\begin{proof}[Proof sketch]
The latent target before rendering is $\prod_{i=1}^n q(s_i)B_0(u_i|s_i)$. {Applying Theorem~\ref{thm:cuff-full} to $(S,U)$ gives~\eqref{eq:ri-region-main}.} The subsequent application of $W_0$ is decoder-local postprocessing and requires no additional rate; representation-induced compatibility ensures that its output law is the prescribed $q(s^n,x^n)$. The converse is the ordinary channel-synthesis converse for the prescribed pair $(S,U)$.
\end{proof}

The theorem shows that once the representation-induced latent law has been specified, {its minimum compression rate is exactly the rate required for synthesizing} $(S,U)$. Whether this latent target realizes the prescribed output law is determined separately by the compatibility condition $(B_0W_0)(x|s)=q(x|s)$, that is, no synthesis code for the fixed latent target can repair a failure of this condition. Whenever compatibility holds,
\begin{align}
    \mathscr R_{\ri}(q|\baseG)
    &\subseteq\mathscr R_{\ro}(q|W_0)
    \subseteq\mathscr R_{\full}(q).
    \label{eq:three-region-inclusion}
\end{align}

\subsubsection{Basic Invariances Under Matched Marginals}

Let $q(x)=q_0(x)$. The aggregate latent law induced by $B_0$ is
\begin{align}
    \sum_sq(s)B_0(u|s)
    &=\sum_xq(x)q_0(u|x) =q_0(u).
    \label{eq:aggregate-latent-preservation}
\end{align}
Thus Representation-Induced Transfer uses the same aggregate latent distribution as the base RDP factorization. Moreover,
\begin{align}
    p(\hat x)
    &=\sum_xq(x)T_0(\hat x|x)
    =q_0(\hat x),
    \label{eq:ri-perception-preservation}
\end{align}
so the marginal output law is preserved exactly at the single-letter level. Under the joint distribution
\begin{align}
    p(s,x,u,\hat x)=q(s,x)q_0(u|x)W_0(\hat x|u),
    \label{eq:ri-joint-law}
\end{align}
the distortion also satisfies
\begin{align}
    \E[d(X,\hat X)]
    &=\sum_{x,\tilde x}q_0(x)T_0(\tilde x|x)d(x,\tilde x)
    \le D_0.
    \label{eq:ri-distortion-preservation}
\end{align}
Hence the transfer preserves the two properties for which the base system was designed---distortion and the unconditional output law---even when it fails to preserve the desired dependence on $S$.

\subsubsection{Exact Transfer and the Recurrent-Class Invariant}

From~\eqref{eq:ri-output-law}, $(B_0W_0)(x'|s) =\sum_xq(x|s)T_0(x'|x)$. Therefore Representation-Induced Transfer is exact if and only if
\begin{align}
    \sum_xq(x|s)T_0(x'|x)
    &=q(x'|s),
    \qquad s\in\suppo(q(s)),
    \label{eq:ri-stationarity}
\end{align}
or equivalently $q(x|s)\in\Fix(T_0)$, {where, for a channel $T$ from $\Xc$ to itself, we define its set of stationary distributions by
\begin{align}
    \Fix(T)
    &:=\{\pi\in\Pc(\Xc):\pi T=\pi\}.
    \label{eq:fix-definition}
\end{align}
}
Thus preservation of the marginal law $q(x)=q_0(x)$ under $T_0$ does not in general imply preservation of the conditional laws $q(x|s)$; exact transfer requires the latter for every $s\in\suppo(q(s))$.

Since every stationary distribution $\pi$ of $T_0$ satisfies
\begin{align}
    \pi(x)
    &=\sum_u\left(\sum_{x'}\pi(x')q_0(u|x')\right)W_0(x|u),
    \label{eq:fix-in-hull}
\end{align}
it is a convex combination of the rows of $W_0$ and hence belongs to $\mathcal C(W_0)$. Thus every exactly transferable target under Representation-Induced Transfer is controllable through $W_0$, and
\begin{align}
    \Fix(T_0)
    &\subseteq\mathcal C(W_0)
    \subseteq\Pc(\Xc).
    \label{eq:three-object-inclusion}
\end{align}

We now give a structural interpretation of~\eqref{eq:ri-stationarity}. View $T_0$ as the transition matrix of a finite-state Markov chain on $\Xsupp$. Since $q_0T_0=q_0$ and $q_0(x)$ is positive on $\Xsupp$, this restriction decomposes into disjoint recurrent classes $\Xc_1,\ldots,\Xc_J$~\cite{Norris1997}, i.e., within each class, every state is reachable from every other state under repeated applications of $T_0$, and no transition leaves the class.
Let
\begin{align}
    C
    &=c(X)\in[1:J]
    \label{eq:class-label}
\end{align}
be the recurrent-class index containing $X$. Thus $C$ is a deterministic function of $X$, determined by $T_0$ up to a relabeling of the classes. Let $\pi_c$ denote the unique stationary distribution within class $c$. Every stationary distribution supported on $\Xsupp$ is a mixture of the $\pi_c$'s~\cite{Norris1997}; in particular,
\begin{align}
    q_0(x)
    &=\sum_{c=1}^Jq_0(c)\pi_c(x),
    &q_0(x|c)
    &=\pi_c(x).
    \label{eq:px-class-mixture}
\end{align}
The channel may randomize the realization of $X$ within a recurrent class, but it never changes the class label $C$. Thus, the class label $C=c(X)$ can itself serve as a semantic source. Taking $S=C$ specifies the task of generating $X$ according to $q_0(x|c)$ given $C=c$.

For two semantic kernels $q(s_1|x)$ and $q(s_2|x)$, write $S_1\preceq_X S_2$ if there exists a channel $p(s_1|s_2)$ such that $q(s_1|x) =\sum_{s_2}q(s_2|x)p(s_1|s_2)$, equivalently $X-S_2-S_1$. Thus $S_1\preceq_XS_2$ means that $S_1$ contains no distinction about $X$ beyond those already contained in $S_2$.

\begin{theorem}[Recurrent-Class Characterization and Maximal Exactly Transferable Semantics]
\label{thm:maximal-invariant}
Let $C=c(X)$ be the recurrent-class label of $T_0$ on $\Xsupp$. Then $C$ itself is exactly transferable under Representation-Induced Transfer. Moreover, for every matched-marginal transfer target $q(s,x)$, the following are equivalent:
\begin{enumerate}[label=(\roman*)]

    \item Representation-Induced Transfer is exact, i.e., $(B_0W_0)(x|s)=q(x|s)$ for every $s\in\suppo(q(s))$;

    \item For every $s\in\suppo(q(s))$, $q(x|s) \in\Fix(T_0)$;
    \item The random variables $(X,C,S)$ satisfy
    \begin{align}
        X-C-S,
        \label{eq:X-C-S}
    \end{align}
    equivalently $q(s|x)=q(s|c(x))$.
\end{enumerate}
Consequently, every exactly transferable semantic variable satisfies $S\preceq_XC$. Hence $C$ is the greatest, and therefore maximal, exactly transferable semantic variable under $\preceq_X$.
\end{theorem}

The proof is given in Appendix~\ref{app:ri-proofs}.
\subsubsection{Semantic Information and Universal Interfaces}

The theorem does not require $S=C$. Every deterministic function or stochastic degradation of $C$ is exactly transferable. Conversely, a task that distinguishes source symbols within the same recurrent class falls outside this exact family. If $T_0$ is irreducible on $\Xsupp$, then $C$ is constant and exact transfer requires $S\perp X$; if $T_0$ is the identity channel, then $C=X$ and every semantic kernel is exactly transferable.

\begin{example}[A Transferable Semantic Variable Different From $C$]
Let $\Xc=\{1,2,3,4,5,6\}$, let $q_0(x)$ be uniform, and let $T_0$ mix uniformly within the three recurrent classes $\Xc_1=\{1,2\}$, $\Xc_2=\{3,4\}$, and $\Xc_3=\{5,6\}$. Thus $C\in\{1,2,3\}$. Define $S =\mathbf 1\{C=3\}$. Although $S\ne C$, we have $X-C-S$, and the two conditional laws of $X$ are mixtures of the classwise stationary distributions; hence $S$ is exactly transferable. In contrast, $S'=\mathbf 1\{X=1\}$ distinguishes two points within $\Xc_1$. Since $\delta_1T_0 =\tfrac12\delta_1+\tfrac12\delta_2 \ne\delta_1$, $S'$ is not exactly transferable.
\end{example}

For a fixed exactly transferable target, $C$ also has a direct synthesis interpretation. Since $C$ is a deterministic function of $X$ and $X-C-S$,
\begin{align}
    I(S;X)
    &=I(S;C)+I(S;X|C) =I(S;C).
    \label{eq:C-rate-interface}
\end{align}
Under unlimited common randomness, choosing the Full Redesign synthesis auxiliary $V=C$ therefore achieves the optimal {compression rate} $I(S;X)$. Thus $C$ is a rate-optimal interface for every fixed exactly transferable task, although a particular task may require only a coarser function of $C$.

\begin{corollary}[Maximum Exactly Transferable Semantic Information]
\label{cor:max-transferable-information}
Let $X\sim q_0$ and $C=c(X)$. The maximum exactly transferable semantic information is
\begin{align}
    \max_{\substack{\Sc,\,q(s|x):\\
        X-C-S}}
    I(X;S)
    &=I(X;C)
    =H(C){,}
    \label{eq:max-transferable-information}
\end{align}
where $\Sc$ ranges over finite alphabets.
\end{corollary}

The result follows from $X-C-S$ and data processing, with equality attained by $S=C$. Thus $H(C)$ is the largest amount of information about $X$ that can be carried by an exactly transferable semantic variable.

Call a deterministic representation $Z=z(X)$ a \emph{universal exact-semantic interface} if
\begin{align}
    X-Z-S
    \label{eq:universal-interface-definition}
\end{align}
for every semantic variable $S$ that is exactly transferable under Representation-Induced Transfer.

\begin{corollary}[Minimal Universal Exact-Semantic Interface]
\label{cor:minimal-universal-interface}
The recurrent-class variable $C$ is a universal exact-semantic interface. Moreover, every other universal exact-semantic interface $Z$ satisfies
\begin{align}
    X-Z-C.
    \label{eq:universal-interface-minimality}
\end{align}
Thus every universal interface must preserve at least the information contained in $C$.
\end{corollary}

The first statement follows from Theorem~\ref{thm:maximal-invariant}. For the second, $C$ itself is exactly transferable, so an interface that supports every exactly transferable task must in particular support $S=C$.

\begin{example}[Why a Universal Interface Must Preserve $C$]
In the three-class example above, consider $S_1=\mathbf 1\{C=1\}$ and $S_2=\mathbf 1\{C=3\}$. Each task is a strict coarsening of $C$ and is exactly transferable. Together, $(S_1,S_2)$ determine all three values of $C$. Therefore a representation intended to support arbitrary future transfer tasks that are exactly transferable cannot merge any two recurrent classes.
\end{example}

\subsubsection{Optimality When Representation-Induced Transfer Is Exact}

For an exactly compatible target, define
\begin{align}
    R_{\ri}^{\infty}(q|\baseG)
    &:=\lim_{R_c\to\infty}\Rri(R_c;q,\baseG).
    \label{eq:ri-unlimited-rate-definition}
\end{align}

\begin{theorem}[Optimality of Exact Representation-Induced Transfer]
\label{thm:ri-optimal}
Suppose Representation-Induced Transfer is exact. Under the Representation-Induced Transfer law~\eqref{eq:ri-joint-law},
\begin{align}
    I(S;U)
    &=I(S;X).
    \label{eq:ri-info-equality}
\end{align}
Consequently,
\begin{align}
    R_{\full}^{\infty}(q)
    &=R_{\ro}^{\infty}(q|W_0)
    =R_{\ri}^{\infty}(q|\baseG)
    =I(S;X),
    \label{eq:ri-optimal-hierarchy}
\end{align}
and the latent law $B_0$ attains the Renderer-Only Transfer minimum.
\end{theorem}
\begin{proof}[Proof sketch]
The law~\eqref{eq:ri-joint-law} satisfies both $S-X-U$ and $S-U-\hat X$. Hence $I(S;U)\le I(S;X)$ and $I(S;U)\ge I(S;\hat X)$. Exact transfer gives $p(s,\hat x)=q(s,x)$, so the two bounds yield~\eqref{eq:ri-info-equality}. The rate equality then follows from Theorem~\ref{thm:unlimited-reuse}, Theorem~\ref{thm:ri-region}, and the Full Redesign benchmark.
\end{proof}

The theorem gives the positive reason to retain the representation-induced latent law. If it already generates the correct conditional law, redesigning semantic steering cannot reduce the minimum {compression rate} under unlimited common randomness; the inherited latent law is then optimal even relative to Full Redesign. This representation-level rate statement should not be confused with the base rate $R_0$ of a particular inherited RDP code, which may be larger.

\subsubsection{Approximate Transfer Guarantees from RDP Distortion}

The base RDP distortion constraint provides an approximation guarantee for the fixed Representation-Induced Transfer law, even when exact transfer fails. This guarantee uses the original single-letter distortion measure $d$ and holds for every transfer target distribution $q(s,x)$ satisfying the matched-marginal condition $q(x)=q_0(x)$.

{
For fixed distributions $p$ and $r$ and a nonnegative cost $d$ on $\Xc\times\Xc$, let $W_d(p,r)$ denote the minimum expected cost over all couplings of $p$ and $r$:
\begin{align*}
    W_d(p,r)
    &:=\min_{\substack{\gamma\ge 0:\ \sum_{x'}\gamma(x,x')=p(x)\\
                       \sum_x\gamma(x,x')=r(x')}}
       \sum_{x,x'}\gamma(x,x')d(x,x').
\end{align*}
}
Both marginals remain fixed in this minimization. In the present setting, they are the desired conditional distribution $q(\cdot|s)$ and the actual output distribution $(B_0W_0)(\cdot|s)$ of the fixed transfer law.

The inherited kernel $T_0$ supplies a feasible coupling
$\gamma_s(x,x')=q(x|s)T_0(x'|x)$.
Consequently,
\begin{align}
    &\sum_sq(s)W_d\bigl(q(\cdot|s),(B_0W_0)(\cdot|s)\bigr)
    \notag\\
    &\quad\le\sum_{s,x,x'}q(s,x)T_0(x'|x)d(x,x')
    \notag\\
    &\quad=\sum_{x,x'}q_0(x)T_0(x'|x)d(x,x')
    \le D_0.
    \label{eq:transport-defect-bound}
\end{align}
Thus, the base distortion budget bounds the average transport cost between the target and actual conditional output distributions, with $T_0$ unchanged.

For Hamming distortion, $W_d$ equals total variation, giving
\begin{align}
    \epsilon_{\ri}(q|\baseG)
    &:=\sum_sq(s)
    \left\|q(\cdot|s)-(B_0W_0)(\cdot|s)\right\|_{\TV}
    \le D_0.
    \label{eq:ri-defect}
\end{align}
Hence, low Hamming distortion in the base system ensures small average conditional TV error after Representation-Induced Transfer, even for targets that do not satisfy the exact-transfer conditions.

For general $d$, the guarantee concerns transport cost; a TV guarantee requires an additional relation between $d$ and Hamming distortion. These bounds average over $q(s)$, rather than controlling every conditional distribution separately. They concern single-letter approximation and do not by themselves establish the vanishing block total variation required for strong synthesis.

\subsection{Full-Model Transfer}
\label{sec:full-model}

Recall that Full-Model Transfer reuses a fixed compatible operational RDP code $\baseC$, with {compression and common-randomness rates} $(R_0,R_{c,0})$ satisfying \eqref{eq:base-rate-r}--\eqref{eq:base-rate-sum}. Both the code and these rates remain unchanged under transfer. {In the equivalent representation~\eqref{eq:fm-induced-law}, matched marginals give $X^n\sim q_0(x^n)$, but this alone does not ensure that the dependence between $S^n$ and the reconstruction is preserved.} With the code and rates fixed, we therefore seek additional conditions on the transfer task under which transfer succeeds. Specifically, we identify a family of tasks supported by every compatible implementation and derive a necessary condition that limits the additional semantic dependence a particular code may preserve at its fixed rate.

\subsubsection{Class Consistency of Compatible Implementations}

The recurrent-class variable $C$ from Theorem~\ref{thm:maximal-invariant} is guaranteed by the nominal factorization $\baseG$. The next lemma connects that single-letter invariant to every compatible operational implementation.

\begin{lemma}[Class Consistency Through the Latent]
\label{lem:class-consistency}
Let $(X,U,\tilde X)$ have the single-letter base law $q_0$, satisfying $X-U-\tilde X$ and $X\overset{d}{=}\tilde X$. Let $c$ denote the recurrent-class label of $T_0$.

For every $u$ with $q_0(u)>0$, there exists a unique class label $c_U(u)$ such that
\begin{align}
    q_0(x|u)>0
    &\Longrightarrow c(x)=c_U(u),
    \label{eq:class-consistency-source}\\
    W_0(y|u)>0
    &\Longrightarrow c(y)=c_U(u).
    \label{eq:class-consistency-output}
\end{align}
{Consequently, $c(X)=c_U(U)=c(\tilde X)$ almost surely.}
\end{lemma}

The proof is given in Appendix~\ref{app:operational-transfer}. Under the ideal law associated with any compatible code $\mathsf C_{0,n}$, the lemma holds coordinatewise. Compatibility then transfers this class-consistency property to the actual block code. For $x^n\in\Xc^n$, write
\begin{align}
    c(x^n)
    &:=\bigl(c(x_1),\ldots,c(x_n)\bigr).
    \label{eq:componentwise-class-map}
\end{align}

\begin{theorem}[Full-Model Transfer: Guaranteed Family and Residual-Semantic Converse]
\label{thm:operational-transfer}
Fix a compatible strong-RDP code sequence $\baseC$ satisfying~\eqref{eq:strong-near-perfect}. Then:
\begin{enumerate}[label=(\roman*)]

    \item {under the base operational law~\eqref{eq:actual-rdp-law}, the inherited implementation preserves the recurrent-class sequence,}
    \begin{align}
        \mathsf P\bigl[c(X^n)\ne c(\tilde X^n)\bigr]
        &\longrightarrow0;
        \label{eq:operational-class-preservation}
    \end{align}

    \item For any matched-marginal task $S$ satisfying
    \begin{align}
        X-C-S,
        \label{eq:future-invariant-task}
    \end{align}
    Full-Model Transfer succeeds using the same fixed code sequence $\baseC$;

    \item if $\baseC$ is full-model transferable for a matched-marginal transfer target $q(s,x)$, then
    {
    \begin{align}
        H(C)+I(S;X|C)
        &\le R_0.
        \label{eq:full-model-converse}
    \end{align}
    }
\end{enumerate}
{Consequently, if $R_0=H(C)$, then Full-Model Transfer is exact if and only if $X-C-S$.}
\end{theorem}

The complete proof is given in Appendix~\ref{app:operational-transfer}.

Comparing the conditions for exact Representation-Induced Transfer and exact Full-Model Transfer under a fixed factorization $\baseG$, we obtain from Theorems~\ref{thm:maximal-invariant} and~\ref{thm:operational-transfer}
\begin{align*}
    \text{exact Representation-Induced Transfer}
    &\iff X-C-S,\\
    X-C-S
    &\implies \text{exact Full-Model Transfer using $\baseC$},
\end{align*}
respectively, for any matched-marginal task and any compatible strong-RDP code sequence $\baseC$.
Thus $C$ specifies a semantic family guaranteed by every compatible code, but it need not describe the full transfer capability of a particular code. A compatible implementation may also preserve distinctions within a recurrent class. The converse limits the resulting additional semantic dependence:

{
\begin{align}
    I(S;X|C)
    &\le R_0-H(C).
    \label{eq:residual-semantic-budget}
\end{align}
}
Thus
\begin{align}
    \bigl\{q(s,x):q(x)=q_0(x),\ I(S;X|C)=0\bigr\}
    &\subseteq
    \bigl\{q(s,x):\text{$\baseC$ is full-model transferable}\bigr\}
    \nonumber\\
    &\subseteq
    {\bigl\{q(s,x):q(x)=q_0(x),\ I(S;X|C)\le R_0-H(C)\bigr\}}.
    \label{eq:full-model-sandwich}
\end{align}
The first term in
\begin{align}
    I(S;X)
    &=I(S;C)+I(S;X|C)
    \label{eq:semantic-information-decomposition}
\end{align}
is dependence represented by the guaranteed invariant, whereas the second term measures dependence on the particular source value within its recurrent class. When $R_0=H(C)$, the converse forces $I(S;X|C)=0$, and the two sides of~\eqref{eq:full-model-sandwich} coincide. Hence $X-C-S$ is again necessary and sufficient for exact Full-Model Transfer. At larger rates, the sufficient condition need not be necessary, as the following example shows.

\begin{example}[Full-Model Transfer Beyond the Recurrent-Class Invariant]
\label{ex:full-model-beyond-C}
Let $X\sim\mathrm{Bern}(1/2)$, $\Uc=\{0,1,e\}$, and
\begin{align*}
    q_0(u|x)
    &=\frac12\ind\{u=x\}+\frac12\ind\{u=e\},\\
    W_0(\tilde x|u)
    &=\begin{cases}
        \ind\{\tilde x=u\},&u\in\{0,1\},\\
        1/2,&u=e,
    \end{cases}
    \qquad \tilde x\in\{0,1\}.
\end{align*}
The nominal channel $T_0$ is a binary symmetric channel with crossover probability $1/4$. It is irreducible, so $C$ is constant and $H(C)=0$. Under Hamming distortion, the factorization has distortion $1/4$ and
$I(X;U)=I(\tilde X;U)=1/2$, and is therefore valid at
$(R_0,R_{c,0},D_0)=(1,0,1/4)$.

For each $n$, take $K_0=1$ and let the $2^n$ codewords $u^n(m_0)$ enumerate all binary sequences. The encoder sends the unique index for which $u^n(m_0)=x^n$. Since both $q_0(x|u)$ and $W_0(\tilde x|u)$ are deterministic for $u\in\{0,1\}$, the actual and ideal distributions coincide. Suppressing the constant $K_0$, they are
\begin{align*}
    p(x^n,m_0,u^n,\tilde x^n)
    &=\bar p(x^n,m_0,u^n,\tilde x^n)\\
    &=2^{-n}\ind\{x^n=u^n=\tilde x^n=u^n(m_0)\}.
\end{align*}
Thus the code sequence is compatible and has perfect perception and zero distortion, satisfying the base constraint $D_0=1/4$.

Now consider the matched-marginal task $S=X$. The induced semantic encoder sends the index of $s^n$, and the inherited renderer outputs $\tilde X^n=s^n$. Hence Full-Model Transfer is exact for every $n$, with $I(S;X|C)=1=R_0-H(C)$. However, $X-C-S$ fails because $C$ is constant, so Representation-Induced Transfer is not exact. The nominal factorization uses the erasure row with probability $1/2$, whereas the operational code uses only the deterministic renderer rows. This is permitted by compatibility: the block encoder need not implement $q_0(u|x)$ symbol by symbol. Here the rate-one code preserves more information than required by the base distortion constraint, and that additional information enables exact transfer of $S=X$.
\end{example}

By the likelihood-encoder construction in~\cite{Wagner2022}, compatible strong-RDP implementations exist at strict versions of the base rate inequalities~\eqref{eq:base-rate-r}--\eqref{eq:base-rate-sum}. Hence the guaranteed family in Theorem~\ref{thm:operational-transfer} is available at the ordinary interior RDP rates, with no task-specific modification of the deployed code.

\section{Examples and Gaussian Specialization}
\label{sec:examples}

\subsection{A Controllable Task With a Positive Price of Transfer}
\label{subsec:overlap-example}

Let $\Xc=\{0,1,2\}$, $q_0(x)=\left(\frac14,\frac12,\frac14\right)$, $\Uc=\{0,1\}$, and
\begin{align}
    W_0(x|0)&=\left(\frac12,\frac12,0\right), \qquad W_0(x|1)=\left(0,\frac12,\frac12\right).
    \label{eq:overlap-renderer}
\end{align}
Define the base representation law by
\begin{align}
    q_0(u=1|x)=\frac{x}{2},
    \qquad x\in\{0,1,2\}.
    \label{eq:overlap-base-encoder}
\end{align}
Then $q_0(u=0)=q_0(u=1)=1/2$, and $\tilde X$ has the same marginal distribution as $X$.  Under Hamming distortion,
\begin{align}
    \E[d(X,\tilde X)]&=\frac12, \qquad I(X;U)=\frac12, \qquad I(\tilde X;U)=\frac12.
    \label{eq:overlap-base-quantities}
\end{align}
Hence this is an $(R_0,R_{c,0},D_0)=(1/2,0,1/2)$ {RDP factorization}.

Let $S\sim\mathrm{Bern}(1/2)$ and choose the transfer target distribution
\begin{align}
    q(x|s=0)&=W_0(x|0), \qquad q(x|s=1)=W_0(x|1).
    \label{eq:overlap-task}
\end{align}
Its output marginal is $q_0(x)$.  The two rows are the endpoints of $\Cren$, so exact Renderer-Only Transfer forces $U=S$.  Therefore
\begin{align}
    R_{\ro}^{\infty}(q|W_0)
    &=1.
    \label{eq:overlap-reuse-rate}
\end{align}
The output symbol $1$ occurs with probability $1/2$ and hides $S$, whereas $0$ and $2$ reveal it.  Thus
\begin{align}
    I(S;X)
    &=\frac12, \qquad \Gamma_{W_0}(q)
    =\frac12.
    \label{eq:overlap-gamma}
\end{align}

The induced test channel is
\begin{align}
    T_0
    &=\begin{bmatrix}
        1/2&1/2&0\\
        1/4&1/2&1/4\\
        0&1/2&1/2
    \end{bmatrix}.
    \label{eq:overlap-T0}
\end{align}
Representation-Induced Transfer changes the two target rows to
\begin{align}
    \left(\frac38,\frac12,\frac18\right),
    \qquad
    \left(\frac18,\frac12,\frac38\right),
    \label{eq:overlap-ri-rows}
\end{align}
respectively, and has defect $\epsilon_{\ri}=1/8$.  This example therefore exhibits
\begin{align}
    \text{Representation-Induced Transfer fails}
    \;<\;
    \text{Renderer-Only Transfer succeeds at a positive price}
    \;<\;
    \text{Full Redesign}.
\end{align}

\subsection{Same RDP Test Channel, Different Renderer-Only Transfer Capability}
\label{subsec:same-T0-example}

Let $\Xc=\{0,1,2\}$, $q_0(x)=\left(\frac14,\frac14,\frac12\right)$, and define the class map
\begin{align}
    {c}(x)
    &=\begin{cases}
        0,&x\in\{0,1\},\\
        1,&x=2.
    \end{cases}
\end{align}
Consider the end-to-end channel
\begin{align}
    T_0(\tilde x|0)
    &=T_0(\tilde x|1)
    =\left(\frac12,\frac12,0\right), \qquad T_0(\tilde x|2)=(0,0,1).
    \label{eq:same-T0-channel}
\end{align}
It preserves $q_0(x)$.  Under class distortion
\begin{align}
    d_{c}(x,y)=\ind\{{c}(x)\ne {c}(y)\},
    \label{eq:class-distortion}
\end{align}
the distortion is zero.

\paragraph{Realization A: class renderer.}
Let $U_A={c}(X)$ and
\begin{align}
    W_A(x|0)&=\left(\frac12,\frac12,0\right), \qquad W_A(x|1)=(0,0,1).
    \label{eq:class-renderer}
\end{align}
Then $I(X;U_A)=I(\tilde X;U_A)=1$. Its renderer hull contains only distributions of the form
\begin{align}
    \left(\frac{\alpha}{2},\frac{\alpha}{2},1-\alpha\right),
    \qquad 0\le\alpha\le1.
    \label{eq:class-renderer-hull}
\end{align}

\paragraph{Realization B: output-latent renderer.}
Let $U_B=\tilde X$, choose $q_0(u_B|x)=T_0(u_B|x)$, and let
\begin{align}
    W_B(y|u)=\ind\{y=u\}.
    \label{eq:identity-renderer}
\end{align}
Then
\begin{align}
    I(X;U_B)=I(X;\tilde X)=1,
    \qquad
    I(\tilde X;U_B)=H(\tilde X)=\frac32.
\end{align}
At the common base point
\begin{align}
    (R_0,R_{c,0},D_0)=\left(1,\frac12,0\right),
    \label{eq:common-base-point}
\end{align}
both are valid RDP factorizations, and $R_0=1$ is minimum because zero class distortion requires the class bit to be recovered {(a proof of this statement is given in Appendix~\ref{app:same-test-channel-separation})}.

\begin{proposition}[RDP-Equivalent Realizations Can Be Renderer-Only-Transfer Inequivalent]
\label{prop:same-T0-separation}
Realizations A and B induce the same $T_0$, the same source--reconstruction joint law, the same Representation-Induced Transfer map, and the same recurrent-class invariant.  Nevertheless,
\begin{align}
    \mathcal C(W_A)
    \subsetneq
    \mathcal C(W_B)=\Pc(\Xc).
    \label{eq:hull-strict-separation}
\end{align}
For the matched-marginal task $S=X$, realization A is uncontrollable, whereas realization B is controllable and has zero Renderer-Only Transfer price at every common-randomness rate.
\end{proposition}

{
For $S=X$, \eqref{eq:class-renderer-hull} excludes the required point masses $\delta_0$ and $\delta_1$, whereas the identity renderer $W_B$ can implement every decoder channel.
}

The example shows that the end-to-end RDP test channel determines Representation-Induced Transfer but not Renderer-Only Transfer.  The internal factorization remains operationally relevant even after the complete classical RDP coupling is fixed.


\subsection{Scalar Gaussian Specialization}
\label{subsec:gaussian}

The preceding general characterizations retain their finite-alphabet scope. Here we specialize our results to the Gaussian case in which the operational definitions use real source, latent, and reconstruction alphabets, probability kernels, and $\|P-P'\|_{\TV}=\sup_A|P(A)-P'(A)|$, where $A$ ranges over measurable sets. Sources and targets remain i.i.d., the renderer remains memoryless, and the base distortion is $d(x,\tilde x)=(x-\tilde x)^2$. The Gaussian base RDP characterization follows from~\cite[Thm.~2 and Prop.~1]{Wagner2022}, whose uniform-integrability condition holds in this setting. For synthesis achievability, we use general-alphabet soft covering~\cite[Secs.~V and VII]{Cuff2013}.

The scalar Gaussian model makes the two representation invariants especially transparent.  The end-to-end test channel determines how the original representation attenuates semantic dependence, whereas the renderer noise determines which semantic conditionals remain reachable after adaptation.  Jointly Gaussian realizations are natural in this setting: scalar Gaussian output-constrained/RDP problems with limited common randomness admit explicit Gaussian characterizations, and Gaussian optimality is known for Gaussian-vector RDP under quadratic criteria~\cite{Wagner2022,XieEtAl2025,QianEtAl2025}.

Let $X_0\sim\N(0,1)$ and consider the jointly Gaussian RDP factorization
\begin{align}
    X_0
    &=\alpha_0 U+\sqrt{1-\alpha_0^2}\,Z_X,
    \label{eq:gauss-base-x}\\
    \tilde X
    &=\beta_0 U+\sqrt{1-\beta_0^2}\,Z_0,
    \label{eq:gauss-base-y}
\end{align}
where $U,Z_X,Z_0$ are independent standard Gaussian random variables and $0<\alpha_0,\beta_0<1$.  The output has the same $\N(0,1)$ marginal as the source.  The end-to-end correlation and squared-error distortion are
\begin{align}
    \rho_0
    &:=\E[X_0\tilde X]=\alpha_0\beta_0,
    \label{eq:gauss-rho0}\\
    D_0
    &=\E[(X_0-\tilde X)^2]=2(1-\rho_0),
    \label{eq:gauss-D0}
\end{align}
and the two base RDP information quantities are
\begin{align}
    I(X_0;U)
    &=\frac12\log\frac{1}{1-\alpha_0^2},
    \label{eq:gauss-base-rate1}\\
    I(\tilde X;U)
    &=\frac12\log\frac{1}{1-\beta_0^2}.
    \label{eq:gauss-base-rate2}
\end{align}
Thus $\beta_0^2$ is the fraction of output variance that can be controlled through the renderer input, whereas $1-\beta_0^2$ is the variance of the renderer-local randomness.

Now let the transfer target pair $(S,X)$ be zero-mean jointly Gaussian with unit variances and correlation $\rho$, so
\begin{align}
    X=\rho S+\sqrt{1-\rho^2}\,Z.
    \label{eq:gauss-semantic-target}
\end{align}
This is a matched-marginal task because $X\sim\N(0,1)$.

\begin{proposition}[Gaussian Controllability and Price of Renderer-Only Transfer]
\label{prop:gaussian-reuse}
For the renderer in~\eqref{eq:gauss-base-y}, the jointly Gaussian transfer target~\eqref{eq:gauss-semantic-target} is representable through $W_0$ if and only if $|\rho|\le\beta_0$.  Under unlimited common randomness, a finite {compression rate} is possible if and only if
\begin{align}
    |\rho|<\beta_0.
    \label{eq:gauss-finite-feasible}
\end{align}
In that case, the unique representing steering law is
\begin{align}
    U|S=s
    \sim
    \N\left(
        \frac{\rho}{\beta_0}s,
        1-\frac{\rho^2}{\beta_0^2}
    \right),
    \label{eq:gauss-steering}
\end{align}
and
\begin{align}
    R_{\ro}^{\infty}(\rho|\beta_0)
    &=\frac12\log\frac{\beta_0^2}{\beta_0^2-\rho^2},
    \label{eq:gauss-reuse-rate}\\
    R_{\full}^{\infty}(\rho)
    &=I(S;X)=\frac12\log\frac{1}{1-\rho^2},
    \label{eq:gauss-full-rate}\\
    \Gamma_{W_0}(\rho)
    &=\frac12\log
      \frac{\beta_0^2(1-\rho^2)}{\beta_0^2-\rho^2}.
    \label{eq:gauss-price}
\end{align}
\end{proposition}

{A proof is given in Appendix~\ref{app:gaussian-reuse}.}

The condition $|\rho|\le\beta_0$ has a direct noise-floor interpretation.  The target conditional variance is $1-\rho^2$, while the fixed renderer already injects variance $1-\beta_0^2$.  Hence the renderer cannot synthesize a conditional law sharper than its own irreducible noise.  As $|\rho|\uparrow\beta_0$, the steering law becomes deterministic and the required mutual information diverges.  Conversely, $\Gamma_{W_0}(\rho)\downarrow0$ as $\beta_0\uparrow1$, i.e., as the renderer approaches a deterministic programmable interface.

Representation-Induced Transfer depends on the product $\rho_0=\alpha_0\beta_0$, not on $\beta_0$ alone.  The induced test channel is
\begin{align}
    \tilde X
    =\rho_0X_0+\sqrt{1-\rho_0^2}\,Z',
    \label{eq:gauss-test-channel}
\end{align}
so applying it to the transfer target pair produces an output pair $(S,\hat X)$ with correlation $\rho\rho_0$. Thus exact nontrivial transfer is possible only in the lossless limit $\rho_0=1$. The other exact case is the independent task $\rho=0$.

Even when exact transfer fails, the induced output can retain dependence on $S$. To quantify this retained dependence, we compare $I(S;\hat X)$ with the target value $I(S;X)=\frac12\log\frac{1}{1-\rho^2}$, which also gives the Full Redesign rate under unlimited common randomness. Under the RI law,
\begin{align}
    I(S;\hat X)
    =\frac12\log\frac{1}{1-\rho^2\rho_0^2},
    \label{eq:gauss-ri-information}
\end{align}
Here the output mutual information measures the semantic dependence retained by the induced single-letter law. Using $\rho_0=1-D_0/2$ from~\eqref{eq:gauss-D0}, we obtain
\begin{align}
    I(S;\hat X)
    =\frac12\log
    \frac{1}{1-\rho^2(1-D_0/2)^2}.
    \label{eq:gauss-distortion-semantic}
\end{align}
For a fixed nontrivial transfer target, a smaller base RDP distortion therefore preserves more mutual information about $S$, with $I(S;\hat X)$ approaching the target value $I(S;X)$ as $D_0\downarrow0$. This quantifies dependence retention, while exact transfer requires preservation of the full transfer target distribution.

\begin{proposition}[Same Gaussian Test Channel, Different Renderer-Only Transfer Capability]
\label{prop:gaussian-same-test}
Consider two Gaussian factorizations of the form~\eqref{eq:gauss-base-x}--\eqref{eq:gauss-base-y} satisfying
\begin{align}
    \alpha_1\beta_1=\alpha_2\beta_2=\rho_0
    \quad\text{and}\quad
    0<\beta_1<\beta_2<1.
    \label{eq:gauss-same-product}
\end{align}
They induce the same end-to-end RDP test channel, the same base source--reconstruction law, and the same Representation-Induced Transfer attenuation.  Nevertheless, any jointly Gaussian transfer target satisfying
\begin{align}
    \beta_1<|\rho|<\beta_2
    \label{eq:gauss-separation-range}
\end{align}
is infeasible through the first renderer and feasible at finite rate through the second.
\end{proposition}

{
The common product in~\eqref{eq:gauss-same-product} gives the same end-to-end RDP test channel by~\eqref{eq:gauss-test-channel}. On the other hand, Proposition~\ref{prop:gaussian-reuse} shows that a transfer target is unrepresentable through renderer $i$ when $|\rho|>\beta_i$, and achievable at finite rate when $|\rho|<\beta_i$. Condition~\eqref{eq:gauss-separation-range} therefore gives the stated strict separation. Thus, the base RDP test channel fixes only the product $\alpha_i\beta_i$ and does not determine the renderer's ability to support future transfer tasks.
}

\subsubsection{Numerical Illustrations}

We next visualize two implications of the Gaussian formulas.  First, consider the {minimum-compression-rate} Gaussian RDP boundary at a fixed base distortion $D_0$ and base common-randomness rate $R_{c,0}$.  Let
\begin{align}
    \rho_0
    &=1-\frac{D_0}{2},
    &c_0
    &=2^{-2R_{c,0}}.
    \label{eq:gauss-plot-parameters}
\end{align}
On the Gaussian RDP boundary characterized in~\cite{Wagner2022}, both information constraints in~\eqref{eq:gauss-base-rate1}--\eqref{eq:gauss-base-rate2} are active.  Combining $\alpha_0\beta_0=\rho_0$ and $1-\beta_0^2=c_0(1-\alpha_0^2)$ gives
\begin{align}
    \beta_0^2
    &=\frac{1-c_0+\sqrt{(1-c_0)^2+4c_0\rho_0^2}}{2}.
    \label{eq:gauss-beta-base-resources}
\end{align}
Since finite-rate Renderer-Only Transfer requires $|\rho|<\beta_0$, define
\begin{align}
    \rho_{\max}(D_0,R_{c,0})
    &:=\sup\left\{|\rho|:
    R_{\ro}^{\infty}(\rho|\beta_0)<\infty\right\}
    =\beta_0.
    \label{eq:gauss-rho-max}
\end{align}

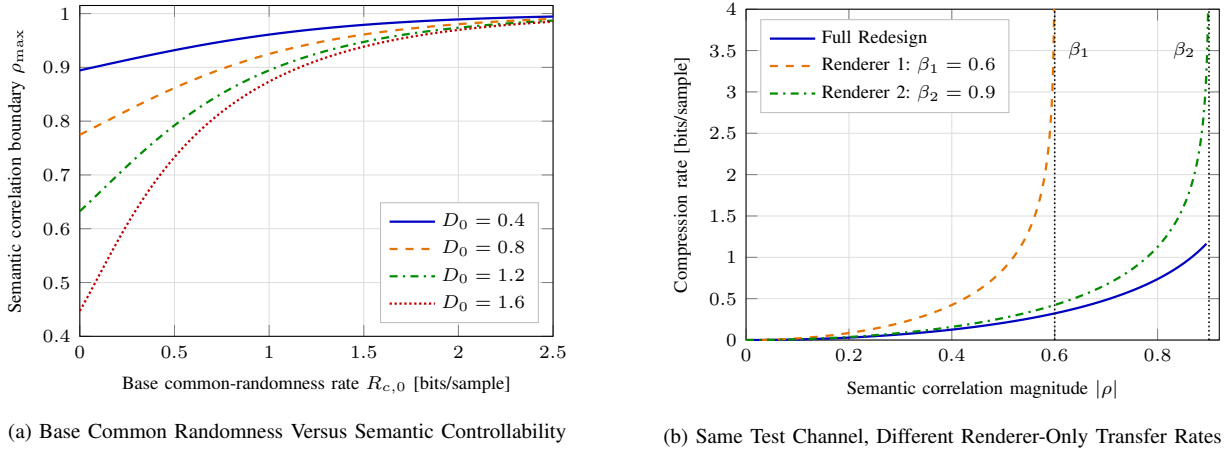
\begin{figure*}[t]
    \centering
    \subfloat[Base Common Randomness Versus Semantic Controllability]{%
        \begin{minipage}{0.475\textwidth}\centering
\begin{tikzpicture}
\begin{axis}[
  width=\linewidth, height=0.76\linewidth,
  xmin=0,xmax=2.5,ymin=0.4,ymax=1.015,
  xlabel={Base common-randomness rate $R_{c,0}$ [bits/sample]},
  ylabel={Semantic correlation boundary $\rho_{\max}$},
  xtick={0,0.5,1,1.5,2,2.5},ytick={0.4,0.5,0.6,0.7,0.8,0.9,1},
  tick label style={font=\scriptsize},label style={font=\scriptsize},
  legend style={font=\scriptsize,at={(0.97,0.04)},anchor=south east,
    draw=gray!50,fill=white,cells={anchor=west}},
  grid=major,grid style={gray!25},axis line style={black},
  every axis plot/.append style={line width=0.85pt,no marks}]
\addplot[blue!75!black] table[x=Rc0,y=D0_04]{numerical_plots/gaussian_transfer/data/controllability.dat};
\addlegendentry{$D_0=0.4$}
\addplot[orange!90!black,dashed] table[x=Rc0,y=D0_08]{numerical_plots/gaussian_transfer/data/controllability.dat};
\addlegendentry{$D_0=0.8$}
\addplot[green!55!black,dashdotted] table[x=Rc0,y=D0_12]{numerical_plots/gaussian_transfer/data/controllability.dat};
\addlegendentry{$D_0=1.2$}
\addplot[red!75!black,densely dotted] table[x=Rc0,y=D0_16]{numerical_plots/gaussian_transfer/data/controllability.dat};
\addlegendentry{$D_0=1.6$}
\end{axis}
\end{tikzpicture}\end{minipage}%
        \label{fig:gaussian-controllability}}
    \hfill
    \subfloat[Same Test Channel, Different Renderer-Only Transfer Rates]{%
        \begin{minipage}{0.475\textwidth}\centering
\begin{tikzpicture}
\begin{axis}[
  width=\linewidth,height=0.76\linewidth,
  xmin=0,xmax=0.92,ymin=0,ymax=4,
  xlabel={Semantic correlation magnitude $|\rho|$},
  ylabel={Compression rate [bits/sample]},
  xtick={0,0.2,0.4,0.6,0.8},ytick={0,0.5,1,1.5,2,2.5,3,3.5,4},
  tick label style={font=\scriptsize},label style={font=\scriptsize},
  legend style={font=\scriptsize,at={(0.03,0.97)},anchor=north west,
    draw=gray!50,fill=white,cells={anchor=west}},
  grid=major,grid style={gray!25},axis line style={black},
  unbounded coords=jump,
  every axis plot/.append style={line width=0.85pt,no marks}]
\addplot[blue!75!black] table[x=rho_full,y=R_full]{numerical_plots/gaussian_transfer/data/transfer_rates.dat};
\addlegendentry{Full Redesign}
\addplot[orange!90!black,dashed] table[x=rho_b1,y=R_b1]{numerical_plots/gaussian_transfer/data/transfer_rates.dat};
\addlegendentry{Renderer 1: $\beta_1=0.6$}
\addplot[green!55!black,dashdotted] table[x=rho_b2,y=R_b2]{numerical_plots/gaussian_transfer/data/transfer_rates.dat};
\addlegendentry{Renderer 2: $\beta_2=0.9$}
\draw[black,densely dotted,line width=0.6pt] (axis cs:0.6,0)--(axis cs:0.6,4);
\draw[black,densely dotted,line width=0.6pt] (axis cs:0.9,0)--(axis cs:0.9,4);
\node[anchor=north west,font=\scriptsize] at (axis cs:0.608,3.72) {$\beta_1$};
\node[anchor=north east,font=\scriptsize] at (axis cs:0.892,3.72) {$\beta_2$};
\end{axis}
\end{tikzpicture}\end{minipage}%
        \label{fig:gaussian-same-test-rates}}
    \caption{Gaussian Illustrations of Renderer-Only Transfer.  Panel~(a) plots the semantic-correlation boundary in~\eqref{eq:gauss-rho-max} along the {minimum-compression-rate} base RDP boundary for $D_0\in\{0.4,0.8,1.2,1.6\}$.  Panel~(b) fixes the end-to-end RDP correlation $\rho_0=0.5$ and compares Full Redesign with two factorizations having $\beta_1=0.6$ and $\beta_2=0.9$; the dotted vertical lines mark their respective controllability boundaries.}
    \label{fig:gaussian-evaluations}
\end{figure*}

Figure~\ref{fig:gaussian-controllability} shows how the resources used to construct the \emph{base} RDP system affect its later semantic capability.  At a fixed $D_0$, increasing $R_{c,0}$ along the {minimum-compression-rate} Gaussian RDP boundary increases $\beta_0$ and therefore enlarges the range $|\rho|<\rho_{\max}$ of jointly Gaussian transfer targets that can be synthesized at finite rate through the inherited renderer.  The horizontal axis is the common-randomness rate of the original RDP system, not the common randomness available to the later semantic code.  Larger base distortion produces a less controllable renderer at $R_{c,0}=0$, while additional base common randomness drives the boundary toward one.

Figure~\ref{fig:gaussian-same-test-rates} isolates the effect of the internal factorization.  The two systems have the same product $\alpha_i\beta_i=\rho_0$, and hence the same end-to-end RDP test channel and base distortion, but the larger renderer coefficient $\beta_2$ yields both a smaller Renderer-Only Transfer rate for every $0<|\rho|<\beta_1$ and a strictly larger feasible task range.  In particular, transfer targets with $0.6<|\rho|<0.9$ are infeasible through the first renderer but feasible through the second.  Thus equal end-to-end RDP behavior can conceal substantially different future semantic controllability.


\section{Concluding Remarks}
\label{sec:discussion}

This work develops an information-theoretic framework for assessing which future transfer tasks can be supported by reusing part or all of a generative RDP system. The results show that this capability cannot be inferred from the original RDP performance alone. In the Gaussian example, different internal factorizations realize the same end-to-end test channel, yet one renderer can support transfer targets that another cannot. Thus, the choice of internal representation and renderer matters beyond its role in the original compression task. This distinction becomes relevant when a system is designed before its eventual transfer tasks are known. If future reuse is anticipated, the initial design could account for both the original compression objective and the class of transfer tasks that the system is expected to support. A natural next problem is to determine how these objectives can be met jointly within a specified class of latent interfaces or renderer architectures.

A complementary question is how much the transfer capability changes when a controlled semantic mismatch is acceptable. For Representation-Induced Transfer, exact synthesis requires each target conditional distribution to be preserved by the inherited kernel. A target outside this invariant class can therefore fail the exact-transfer condition even when the actual conditional output distribution closely approximates the desired one. The base distortion bound provides a first guarantee in this direction by controlling the average single-letter transport cost between these distributions. This guarantee concerns the fixed transfer law, however, and does not by itself establish the vanishing block total variation required for strong synthesis. An operational theory of approximate transfer would connect these two levels: under a specified allowance for semantic mismatch, which additional transfer tasks become feasible, and what {compression and common-randomness rates} do they require?

\appendices

\section{Proof of the Renderer-Only Transfer Region}
\label{app:reuse-region}

\subsection{Achievability}

Fix a joint distribution
\begin{align}
    p(s,v,u,x)=q(s)p(v|s)p(u|v)W_0(x|u)
    \label{eq:app-reuse-law}
\end{align}
whose $(S,X)$ marginal equals $q(s,x)$, and define
\begin{align}
    \overline W_0(x|v)
    :=\sum_up(u|v)W_0(x|u).
    \label{eq:app-effective-channel}
\end{align}
Then $S-V-X$ and $q(s,x)=\sum_vq(s)p(v|s)\overline W_0(x|v)$. By Theorem~\ref{thm:cuff-full}, rates satisfying $R>I(S;V)$ and $R+R_c>I(S,X;V)$ permit strong synthesis using the effective decoder $\overline W_0$.  The decoder implements each use of $\overline W_0$ by drawing $U_i$ according to $p(u_i|v_i)$ given $V_i=v_i$ and then applying $W_0(x_i|u_i)$.  This realizes exactly the required operational form and proves achievability after closure.

\subsection{Converse}

We use the single-letterization argument for distributed channel synthesis~\cite[Sec.~VI]{Cuff2013}, while retaining the restriction that the decoder operates through $W_0$.

Consider any sequence of Renderer-Only Transfer codes satisfying~\eqref{eq:reuse-criterion}. Let $M$ and $K$ be the message and common-randomness index, and let $U^n$ be the decoder-generated renderer input. Let $Q\sim\U[1:n]$ be independent of all code variables, and define
\begin{align}
    V=(M,K,Q),
    \qquad
    S=S_Q,
    \quad U=U_Q,
    \quad X=\hat X_Q.
    \label{eq:app-converse-vars}
\end{align}
Since $U^n$ is generated from $(M,K)$ using decoder-private randomness and the renderer remains memoryless, the conditional distribution factors as
\begin{align}
    p(s,u,x|v)=p(s|v)p(u|v)W_0(x|u).
    \label{eq:app-renderer-factorization}
\end{align}
Thus $S-V-U-X$, with the prescribed renderer $W_0$.

For the {compression rate},

\begin{align}
    nR
    &\ge H(M|K)
    \ge I(S^n;M|K)
    \nonumber\\
    &\stackrel{(a)}{=}I(S^n;M,K)
    \nonumber\\
    &\stackrel{(b)}{\ge}\sum_{i=1}^n I(S_i;M,K)
    \nonumber\\
    &=nI(S_Q;M,K|Q)
    \stackrel{(c)}{=}nI(S;V),
    \label{eq:app-R-bound}
\end{align}
where step $(a)$ follows from the fact that $S^n$ is independent of $K$, step $(b)$ follows from the independence of the source coordinates and the fact that conditioning reduces entropy, and step $(c)$ follows from $S_Q\perp Q$ and the definition $V=(M,K,Q)$.

For the sum rate, the synthesized pairs need not be exactly i.i.d. By strong synthesis and~\cite[Lemma~VI.3]{Cuff2013}, there is a sequence $\delta_n\to0$ such that
\begin{align}
    \frac1n\sum_{i=1}^n I(S_i,\hat X_i;S^{i-1},\hat X^{i-1})
    &\le\delta_n,
    \nonumber\\
    I(S_Q,\hat X_Q;Q)&\le\delta_n.
    \label{eq:app-TC-vanish}
\end{align}
These bounds control the dependence across coordinates and the dependence on the selected time index, respectively. The chain rule now gives
\begin{align}
    n(R+R_c)
    &\ge H(M,K)
    \ge I(S^n,\hat X^n;M,K)
    \nonumber\\
    &=\sum_{i=1}^n I(S_i,\hat X_i;M,K,S^{i-1},\hat X^{i-1})
    \nonumber\\
    &\quad-\sum_{i=1}^n I(S_i,\hat X_i;S^{i-1},\hat X^{i-1})
    \nonumber\\
    &\stackrel{(a)}{\ge}\sum_{i=1}^n I(S_i,\hat X_i;M,K)-n\delta_n
    \nonumber\\
    &=nI(S_Q,\hat X_Q;M,K|Q)-n\delta_n
    \nonumber\\
    &=nI(S,X;V)-nI(S_Q,\hat X_Q;Q)-n\delta_n
    \nonumber\\
    &\stackrel{(b)}{\ge}nI(S,X;V)-2n\delta_n,
    \label{eq:app-sum-single}
\end{align}

where step $(a)$ follows from dropping the past coordinates from the second argument of mutual information and applying the first bound in~\eqref{eq:app-TC-vanish}, and step $(b)$ follows from its second bound.

Strong synthesis also gives $p(s,x)\to q(s,x)$ by marginalization and averaging over $Q$. Using the cardinality bound established below, a convergent subsequence of $p(v)$, $p(s|v)$ and $p(u|v)$ yields a limiting distribution satisfying~\eqref{eq:app-renderer-factorization} and the exact target marginal. Finite-alphabet continuity of mutual information then gives a law in~\eqref{eq:reuse-family} satisfying both rate bounds, proving the converse.

Finally, by the support lemma~\cite[Appendix~C]{ElGamalKim2011}, it suffices to take $|\Vc|\le|\Sc||\Xc|+1$ while preserving $p(s,x)$, both rate quantities, and the factorization in~\eqref{eq:app-renderer-factorization}. {With this cardinality bound, compactness of the admissible family and continuity of mutual information imply that the union in~\eqref{eq:reuse-region} is closed.}

\section{Proof of Theorem~\ref{thm:universal-no-loss}}
\label{app:universal-no-loss}

We prove the equivalence through the cycle $(i) \Longrightarrow(ii) \Longrightarrow(iii) \Longrightarrow(iv) \Longrightarrow(i)$.

For $(i)\Rightarrow(ii)$, consider the matched-marginal transfer target obtained by taking $S=X$, $q(s)=q_0(s)$, and $q(x'|s=x)=\delta_x(x')$. By controllability and Corollary~\ref{cor:controllability}, $\delta_x \in\Cren, \quad x\in\Xsupp$. Since $\Cren$ is convex, $\Pc(\Xsupp) =\conv\{\delta_x:x\in\Xsupp\} \subseteq\Cren$, which proves $(ii)$.

For $(ii)\Rightarrow(iii)$, since $\delta_x\in\Pc(\Xsupp)\subseteq\Cren$ for every $x\in\Xsupp$, and $\delta_x$ is an extreme point of the probability simplex, there must exist $u_x\in\Uc$ such that $W_0(x'|u_x)=\delta_x(x')$.

For $(iii)\Rightarrow(iv)$, fix an arbitrary matched-marginal transfer target $q(s,x)$ and any Full Redesign law $p(s,v,x)$ inducing it with $S-V-X$. By (iii), choose an injective map $g:\Xsupp\to\Uc$ such that $W_0(\cdot|g(x))=\delta_x$ for every $x\in\Xsupp$. Since $X$ is supported on $\Xsupp$, define $U=g(X)$. Then $X=g^{-1}(U)$, and $S-V-U-X$ with the prescribed renderer $W_0$.

This construction preserves the joint law of $(S,V,X)$, and hence both $I(S;V)$ and $I(S,X;V)$. Therefore, $\mathscr R_{\full}(q) \subseteq\mathscr R_{\ro}(q|W_0)$. The reverse inclusion is obvious by definition. This proves $(iv)$.

Finally, for $(iv)\Rightarrow(i)$, the Full Redesign region is nonempty for every finite-alphabet transfer target. Hence~\eqref{eq:universal-region-equality} implies $\mathscr R_{\ro}(q|W_0)\ne\emptyset$ for every matched-marginal transfer target, which is equivalent to controllability by Corollary~\ref{cor:controllability}.

\section{Proof of Theorem~\ref{thm:maximal-invariant}}
\label{app:ri-proofs}

The equivalence of (i) and (ii) follows from~\eqref{eq:ri-output-law}. We prove the equivalence of (ii) and (iii) and then establish maximality.

We use the finite-state recurrence and stationary-distribution results in~\cite[Secs.~1.5 and 1.7]{Norris1997}. Since $q_0(x)$ is stationary and positive on $\Xsupp$, this support is closed and contains no transient states. Hence the restriction of $T_0$ to $\Xsupp$ is a disjoint union of closed irreducible recurrent classes. Each class $c$ has a unique stationary distribution $\pi_c$. Conditioning any stationary distribution on a class of positive mass gives $\pi_c$, so every stationary distribution on $\Xsupp$ is a convex combination of these class distributions. In particular,
\begin{align}
    q_0(x|c)
    &=\pi_c(x).
    \label{eq:app-px-class}
\end{align}

Assume (ii). For each $s\in\suppo(q(s))$,
\begin{align}
    q(x|s)
    &=\sum_c q(c|s)\pi_c(x).
    \label{eq:app-row-stationary-mixture}
\end{align}
Together with~\eqref{eq:app-px-class}, this gives $q(x|s,c) =q_0(x|c)$, and hence $X-C-S$.

Assume (iii). Then
\begin{align}
    q(x|s)
    &=\sum_cq(c|s)q_0(x|c) =\sum_cq(c|s)\pi_c(x),
    \label{eq:app-converse-class-mixture}
\end{align}
which is stationary under $T_0$, proving (ii).

Finally, take $S=C$. By~\eqref{eq:app-px-class}, every conditional law $q_0(x|c)=\pi_c(x)$ is stationary, so $C$ itself is exactly transferable. Every other exactly transferable $S$ satisfies $X-C-S$, i.e., $S\preceq_XC$. Hence $C$ is the greatest exactly transferable semantic variable.

\section{Proofs for Full-Model Transfer}
\label{app:operational-transfer}

\subsection{Proof of Lemma~\ref{lem:class-consistency}}

Fix $u\in\Uc$, any $x$ such that $q_0(x|u)>0$, and any $y$ such that $W_0(y|u)>0$. Since $q_0(u)>0$, Bayes' rule gives $q_0(x)q_0(u|x)>0$, and therefore
\begin{align}
    T_0(y|x)
    &\ge q_0(u|x)W_0(y|u)
    >0.
    \label{eq:app-positive-transition}
\end{align}
Because $q_0T_0=q_0$, the support $\Xsupp$ is closed under $T_0$, and every state in $\Xsupp$ belongs to a closed recurrent class. A positive-probability transition from $x$ to $y$ therefore keeps both states in the same recurrent class. Since $x$ and $y$ were arbitrary in the two conditional supports associated with $u$, both supports lie in one class depending only on $u$. {This proves Lemma~\ref{lem:class-consistency}.}

\subsection{Proof of Theorem~\ref{thm:operational-transfer}}

{
We first transfer class preservation $c(X^n)=c(\tilde X^n)$ from the ideal distribution to the actual code. We then use that preservation to establish the sufficient task condition, and finally use the message-entropy bound to prove necessity.
}

{Let $p$ and $\bar p$ denote the actual and ideal laws from~\eqref{eq:actual-rdp-law} and~\eqref{eq:ideal-rdp-law}, respectively.} Compatibility gives
\begin{align}
    \|p-\bar p\|_{\TV}
    &\longrightarrow0.
    \label{eq:app-compatible-tv}
\end{align}
Under $\bar p$, conditional on a selected codeword, the coordinates of $X^n$ and $\tilde X^n$ are generated according to $q_0(x|u)$ and $W_0$. Lemma~\ref{lem:class-consistency} therefore gives $c(X^n) =c(\tilde X^n) \quad \bar p\text{-a.s.}$ Let $\bar{\mathsf P}$ denote event probability under the ideal law $\bar p$. For the event $\mathcal E_n=\{c(X^n)\ne c(\tilde X^n)\}$,
\begin{align}
    \mathsf P(\mathcal E_n)
    &\le \bar{\mathsf P}(\mathcal E_n)+\|p-\bar p\|_{\TV}
    \longrightarrow0,
\end{align}
which proves part (i).

For part (ii), fix any matched-marginal transfer target satisfying $q(s|x)=q(s|c(x))$.
{
We use the equivalent representation~\eqref{eq:fm-induced-law} of the induced encoder~\eqref{eq:fm-encoder}. Since $q(s^n,x^n)=q_0(x^n)q(s^n|c(x^n))$, this law can be obtained by appending the product kernel
} $q(s^n|c(x^n)) =\prod_{i=1}^nq(s_i|c(x_i))$ {
to the base operational law~\eqref{eq:actual-rdp-law}. In particular, its semantic source--reconstruction marginal is
\[
    p(s^n,\tilde x^n)
    =\sum_{x^n}p(x^n,\tilde x^n)q(s^n|c(x^n)).
\]
We now introduce a comparison distribution that keeps the actual reconstruction marginal $p(\tilde x^n)$ and replaces the input class $c(x^n)$ in the semantic kernel by the reconstruction class $c(\tilde x^n)$:
}
\begin{align}
    \tilde p(s^n,\tilde x^n)
    :=p(\tilde x^n)
      q(s^n|c(\tilde x^n)).
\end{align}
{
The two semantic kernels agree whenever $c(X^n)=c(\tilde X^n)$. Coupling them through the same pair $(X^n,\tilde X^n)$ therefore gives
}
\begin{align}
    \|p(s^n,\tilde x^n)-\tilde p\|_{\TV}
    &\le \mathsf P(\mathcal E_n),
\end{align}
{
while the target law equals $q_0(\tilde x^n)q(s^n|c(\tilde x^n))$ and thus uses the same semantic kernel as $\tilde p$. Since the reconstruction is retained, their TV distance equals that between their reconstruction marginals:
}
\begin{align}
    \|\tilde p-\prod_{i=1}^n q(s_i,\tilde x_i)\|_{\TV}
    &=\|p(\tilde x^n)-q_0(x^n)\|_{\TV}.
\end{align}
{
The first term tends to zero by part (i), and the second by the base code's strong near-perfect perception property~\eqref{eq:strong-near-perfect}. The triangle inequality therefore proves part (ii).
}

For part (iii), suppose $\baseC$ is full-model transferable for $q(s,x)$. {Under~\eqref{eq:fm-induced-law}, part (i) gives}
\begin{align}
    \eta_n
    :=\mathsf P[c(X^n)\ne c(\tilde X^n)]
    &\longrightarrow0,
\end{align}
and let
\begin{align}
    \epsilon_n
    :=\|p(s^n,\tilde x^n)-\prod_{i=1}^n q(s_i,\tilde x_i)\|_{\TV}
    &\longrightarrow0.
\end{align}
Applying the deterministic map $\tilde x^n\mapsto c(\tilde x^n)$ to the Full-Model output law and then replacing $c(\tilde X^n)$ by $c(X^n)$ gives
\begin{align}
&\left\|
p(s^n,c^n,\tilde x^n)
-\prod_{i=1}^n q(s_i,c_i,\tilde x_i)
\right\|_{\TV}\le \epsilon_n+\eta_n
\longrightarrow0.
\label{eq:app-fm-joint-close}
\end{align}
{
Applying finite-alphabet entropy continuity to the joint distribution in~\eqref{eq:app-fm-joint-close} and its two marginals (see~\cite[eq.~(73)]{Cuff2013}) therefore yields
}
\begin{align}
    \frac1n I(S^n,c(X^n);\tilde X^n)
    &\longrightarrow I(S,C;X).
    \label{eq:app-fm-mi-limit}
\end{align}
{The factorization in~\eqref{eq:fm-induced-law} gives} $(S^n,c(X^n))-(M_0,K_0)-\tilde X^n$, and $K_0$ is independent of $(S^n,c(X^n))$. {Hence
\begin{align}
    nR_0
    &\ge H(M_0|K_0)
    \nonumber\\
    &\ge I(S^n,c(X^n);M_0|K_0)
    \nonumber\\
    &=I(S^n,c(X^n);M_0,K_0)
    \nonumber\\
    &\ge I(S^n,c(X^n);\tilde X^n),
    \label{eq:app-fm-rate-bound}
\end{align}
{where the first inequality follows from $H(M_0|K_0)\le nR_0$, since $M_0\in[1:2^{nR_0}]$.}
Combining~\eqref{eq:app-fm-mi-limit}--\eqref{eq:app-fm-rate-bound} and letting $n\to\infty$ gives
\begin{align}
    R_0
    &\ge I(S,C;X) =H(C)+I(S;X|C).
\end{align}
This proves~\eqref{eq:full-model-converse}. If $R_0=H(C)$, the converse forces $I(S;X|C)=0$, equivalently $X-C-S$, while part (ii) proves achievability.
}

\section{{Optimality of the Common Base Rate}}
\label{app:examples}
\label{app:same-test-channel-separation}

{
To prove the optimality claim following~\eqref{eq:common-base-point}, consider any operational code whose average class distortion tends to zero, and let
\[
    \delta_n:=\frac1n\sum_{i=1}^n
      \mathsf P[{c}(X_i)\ne {c}(\tilde X_i)]\longrightarrow0.
\]
Writing $h_2$ for binary entropy, coordinatewise Fano's inequality and concavity give
\[
    H({c}(X^n)|{c}(\tilde X^n))
    \le\sum_{i=1}^n H({c}(X_i)|{c}(\tilde X_i))
    \le n h_2(\delta_n)=o(n).
\]
Since the class bits are independent and uniform, $K\perp {c}(X^n)$, and ${c}(X^n)-(M,K)-{c}(\tilde X^n)$, we obtain
\begin{align}
    nR_0
    &\ge H(M|K)
    \ge I({c}(X^n);M,K)
    \nonumber\\
    &\ge I({c}(X^n);{c}(\tilde X^n))
    \nonumber\\
    &\ge n-nh_2(\delta_n).
\end{align}
Thus $R_0\ge1$ in the limit, proving that the {common compression rate} $R_0=1$ is optimal.
}

\section{{Proof of Proposition~\ref{prop:gaussian-reuse}}}
\label{app:gaussian}
\label{app:gaussian-reuse}

{
Under the fixed renderer~\eqref{eq:gauss-base-y}, matching the conditional variance of the target in~\eqref{eq:gauss-semantic-target} requires
\[
    \beta_0^2\Var(U|S=s)
    =(1-\rho^2)-(1-\beta_0^2)
    =\beta_0^2-\rho^2.
\]
Thus $|\rho|\le\beta_0$ is necessary for representability. Conversely, the steering law~\eqref{eq:gauss-steering} produces the required Gaussian conditional distribution whenever $|\rho|\le\beta_0$, with zero conditional variance allowed at equality. This representing law is unique because convolution with Gaussian noise is injective, as its characteristic function is everywhere nonzero.

For $|\rho|<\beta_0$, the standard Gaussian mutual-information formula gives $I(S;U)$ equal to the expression in~\eqref{eq:gauss-reuse-rate}. At $|\rho|=\beta_0$, the steering law reduces to $U=(\rho/\beta_0)S$ almost surely, so $I(S;U)=\infty$. We next establish the corresponding operational lower bound for any sequence of codes satisfying strong synthesis.
}

For the converse, consider any sequence of Renderer-Only codes satisfying~\eqref{eq:reuse-criterion}, with message alphabets $\mathcal M_n$, and write $R_n=n^{-1}\log|\mathcal M_n|$ for its {compression rate}. Let $Q$ be uniform on $[1:n]$ and independent of all code variables. Since $K\perp S^n$, $S^n-(M,K)-U^n$, and the source coordinates are independent,
\begin{align*}
    nR_n
    &\ge I(S^n;M,K)
    \ge I(S^n;U^n)
    \ge \sum_{i=1}^n I(S_i;U_i)
    =nI(S_{Q};U_{Q}|Q)
    \ge nI(S_{Q};U_{Q}),
\end{align*}
where the last inequality uses $S_{Q}\perp Q$.

The synthesis criterion implies that $(S_{Q},\hat X_{Q})$ converges in total variation to the transfer target pair $(S,X)$. Thus, for every $a,b\in\mathbb R$, the fixed Gaussian renderer gives
\begin{align*}
    \varphi_{S_{Q},U_{Q}}(a,b)
    &=\varphi_{S_{Q},\hat X_{Q}}(a,b/\beta_0)
      \exp\!\left(\frac{1-\beta_0^2}{2\beta_0^2}b^2\right)\\
    &\longrightarrow
      \exp\!\left[-\frac12
      \left(a^2+b^2+2\frac{\rho}{\beta_0}ab\right)\right].
\end{align*}
If $|\rho|>\beta_0$, the limit exceeds one at $a=-(\rho/\beta_0)b$ with $b\ne0$, contradicting the modulus bound for characteristic functions. Hence no such code sequence exists in that range.

For $|\rho|\le\beta_0$, the characteristic-function convergence theorem~\cite[Thm.~3.10.5]{Durrett2019} implies that $(S_{Q},U_{Q})$ converges weakly to the jointly Gaussian pair $(S,U)$ with unit variances and correlation $\rho/\beta_0$. Joint lower semicontinuity of relative entropy under weak convergence~\cite[Thm.~19]{VanErvenHarremoes2014}, applied to the joint distribution and the product of its marginals, then gives
$\liminf_{n\to\infty}R_n\ge\liminf_{n\to\infty}I(S_{Q};U_{Q})\ge I(S;U)$.
Together with the preceding Gaussian calculation, this proves the lower bound in~\eqref{eq:gauss-reuse-rate} and excludes finite rate when $|\rho|=\beta_0$.

For $|\rho|<\beta_0$, apply the general-alphabet synthesis construction in~\cite[Secs.~V and VII]{Cuff2013} to $S-U-\hat X$, using $U$ as the codebook variable and $W_0$ as the decoder kernel. Strong synthesis is achievable whenever $R>I(S;U)$ and $R+R_c>I(S,\hat X;U)$. The Gaussian law is nonsingular in this range, so both mutual informations are finite. Hence every $R>I(S;U)$ is achievable with a sufficiently large finite common-randomness rate, proving the upper bound in~\eqref{eq:gauss-reuse-rate}.

Equation~\eqref{eq:gauss-full-rate} follows by letting the common-randomness rate tend to infinity in the Gaussian TV-approximate channel-synthesis region of~\cite[Cor.~3]{YuTan2020}. Subtracting it from~\eqref{eq:gauss-reuse-rate} gives~\eqref{eq:gauss-price}.


\end{document}